\documentclass[twoside,leqno,twocolumn]{article}
\pdfoutput=1
\usepackage[letterpaper]{geometry}

\usepackage{ltexpprt}

\usepackage[%
  activate={true,nocompatibility},%
  final,%
  tracking=true,%
  kerning=true,%
  spacing=true,%
  stretch=10,%
  shrink=30]{microtype}
\microtypecontext{spacing=nonfrench}

\usepackage[ruled]{algorithm}
\usepackage{xcolor}
\usepackage{xspace}
\usepackage{booktabs}
\usepackage{dsfont}
\usepackage{footmisc}
\usepackage{marvosym}
\usepackage{pgfplots}
\pgfplotsset{compat=newest}

\usepgfplotslibrary{groupplots}
\pgfplotsset{every axis/.style={scale only axis}}

\pgfplotsset{
  major grid style={thin,dotted},
  minor grid style={thin,dotted},
  ymajorgrids=true,
  yminorgrids=true,
  every axis/.append style={
    line width=0.7pt,
    tick style={
      line cap=round,
      thin,
      major tick length=4pt,
      minor tick length=2pt,
    },
  },
  legend cell align=left,
  legend style={
    line width=0.7pt,
    /tikz/every even column/.append style={column sep=3mm,black},
    /tikz/every odd column/.append style={black},
  },
  legend style={font=\small},
  title style={yshift=-2pt},
  enlarge x limits=0.04,
  every tick label/.append style={font=\footnotesize},
  every axis label/.append style={font=\small},
  every axis y label/.append style={yshift=-1ex},
  /pgf/number format/1000 sep={},
  axis lines*=left,
  xlabel near ticks,
  ylabel near ticks,
  axis lines*=left,
  label style={font=\footnotesize},
  tick label style={font=\footnotesize},
  plotMaximumLoadFactor/.style={
    width=36.0mm,
    height=35.0mm,
  },
}

\usetikzlibrary{pgfplots.statistics}
\usepackage{pgfplots}
\usepackage{amsmath,amssymb}
\usepackage{hyperref}
\usepackage[capitalise,noabbrev]{cleveref}
\usepackage{tabularx}
\usepackage{listings}
\usepackage{multirow}
\usepackage{subcaption}
\usepackage[switch]{lineno}
\usepackage{newunicodechar}
\newunicodechar{₁}{_1}
\newunicodechar{₂}{_2}
\newunicodechar{₃}{_3}
\newunicodechar{λ}{\lambda}
\newunicodechar{∈}{\in}
\newunicodechar{∄}{\nexists}
\def\e{\mathrm{e}}  

\crefname{listing}{Algorithm}{Algorithms}

\let\oldcite\cite
\renewcommand\cite{\unskip~\oldcite}

\newcommand{\myparagraph}[1]{\subparagraph*{#1}}
\newcommand{\Oh}[1]{\mathrm{O}\!\left( #1\right)}
\newcommand{\Om}[1]{\Omega\!\left( #1\right)}

\newcommand{\mytitle}{SicHash --- Small Irregular Cuckoo Tables for Perfect Hashing}
\title{\mytitle}

\newcommand{\mythanks}[3]{\thanks{#1 \href{mailto:#2}{\includegraphics[height=8pt]{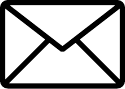}} \href{https://orcid.org/#3}{\includegraphics[height=8pt]{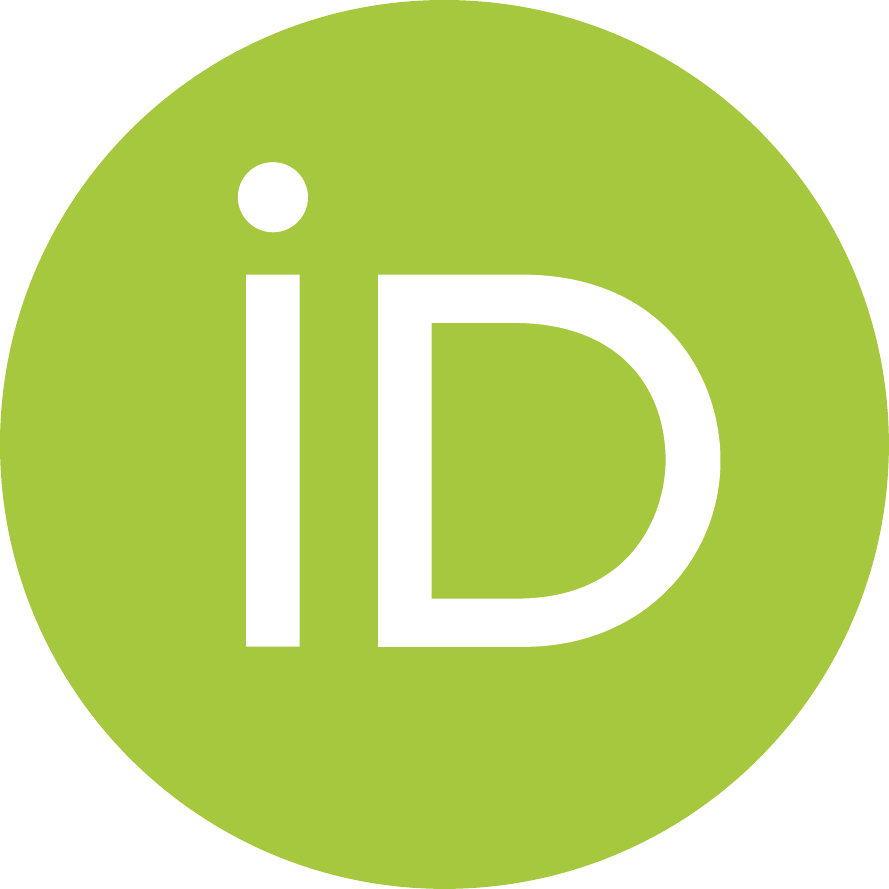}}}}
\date{}
\author{
  Hans-Peter Lehmann\mythanks{Karlsruhe Institute of Technology, Germany.}{hans-peter.lehmann@kit.edu}{0000-0002-0474-1805}
  \and
  Peter Sanders\mythanks{Karlsruhe Institute of Technology, Germany.}{sanders@kit.edu}{0000-0003-3330-9349}
  \and
  Stefan Walzer\mythanks{Cologne University, Germany.}{walzer@cs.uni-koeln.de}{0000-0002-6477-0106}
}

\hypersetup{
  colorlinks=true,
  pdftitle={\mytitle},
  pdfauthor={Hans-Peter Lehmann, Peter Sanders and Stefan Walzer},
  pdfsubject={}
}

\begin{document}
\maketitle

\begin{abstract}
  A Perfect Hash Function (PHF) is a hash function that has no collisions on a given input set.
  PHFs can be used for space efficient storage of data in an array, or for determining a compact representative of each object in the set.
  In this paper, we present the PHF construction algorithm SicHash --- Small Irregular Cuckoo Tables for Perfect Hashing.
  At its core, SicHash uses a known technique:
  It places objects in a cuckoo hash table and then stores the final hash function choice of each object in a retrieval data structure.
  We combine the idea with \emph{irregular} cuckoo hashing, where each object has a different number of hash functions.
  Additionally, we use many small tables that we \emph{overload} beyond their asymptotic maximum load factor.
  The most space efficient competitors often use brute force methods to determine the PHFs.
  SicHash provides a more direct construction algorithm that only rarely needs to re-compute parts.
  Our implementation improves the state of the art in terms of space usage versus construction time for a wide range of configurations.
  At the same time, it provides very fast queries.
\end{abstract}

\section{Introduction}
A Perfect Hash Function (PHF) is a hash function that does not have collisions on a given set $S$ of objects, i.e., is injective.
In this paper, we call the number of objects $N$ and the size of the output range $M$.
A PHF with $M=N$ is called \emph{Minimal} Perfect Hash Function (MPHF).
For $M$ close to $N$, it is likely that an ordinary hash function has collisions.
It is therefore necessary to store additional information, specific to the set $S$, that help to avoid these collisions.
The lower space bound for an MPHF is $1.44$ bits/object \cite{belazzougui2009hash}.
PHFs can be represented with less space, depending on the load factor $N/M$.
In the literature, load factors between $0.8$ and $1.0$ are common \cite{belazzougui2009hash, pibiri2021pthash}.

PHFs can be used to implement a hash table where each query only needs to access one single cell.
When we know that the hash table is never queried for objects $\not\in S$, table cells can store the plain payload data without keys.
This results in a retrieval data structure that can be updated efficiently.
Refer to \cref{sec:prelim} for more details about retrieval data structures and an introduction to cuckoo hashing, which is a main building block of SicHash.
There is a wide range of PHF construction algorithms, which we review extensively in \cref{sec:related}.

The basic idea of SicHash is to distribute the input objects to a number of small buckets and build a cuckoo hash table in each.
To obtain a PHF, we store which hash function index was finally used to place each object by using a retrieval data structure.
The reason for constructing \emph{small} tables is the use of overloading, which we describe in \cref{sec:overloading}. 
In \cref{sec:theDataStructure}, we explain SicHash in detail, giving construction and query algorithms.
Enhancements of the basic scheme, including one that produces MPHFs, are given in \cref{sec:enhancements}.
We analyze SicHash in \cref{sec:analysis}.
We then provide an extensive experimental evaluation in \cref{sec:experiments}, comparing it with a wide range of competitors.
In \cref{sec:conclusion}, we summarize the results and discuss possible future work.

\myparagraph{Comparison with Previous Approaches.}
The most space-efficient previous algorithms perform brute-force search as a core step to determine a perfect hash function.
In particular, both CHD \cite{belazzougui2009hash} and RecSplit \cite{esposito2020recsplit}
at some point try out random hash functions (on subproblems) until one happens to be injective.
Refer to \cref{sec:related} for details.
Our construction is more directed than this because it constructs cuckoo hash tables as its base case, which is possible in polynomial time.
The directedness is also visible in the experiments, where our method can construct PHFs with the same space requirements significantly faster than the competitors.

\myparagraph{Our Contribution.}
We combine and refine several known ideas in a novel way leading to excellent space--construction time trade-offs  while using very low query time.
We base SicHash on the known idea of PHF generation through cuckoo hashing.
We use \emph{irregular} cuckoo hashing, which was previously considered for reducing search time \cite{dietzfelbinger2010tight}.
For that application it was of little help apart from allowing to interpolate between two uniform degree cases.
In contrast, for our application to reduce space, it is helpful even for integer average degree.
Space is further reduced using the novel idea to \emph{overload}  the cuckoo hash tables, i.e., to load them with more objects than would be possible in an asymptotic sense, exploiting that the tables are small.
All this keeps the queries extremely simple --- basically the cost for a single access to a retrieval data structure.
This further profits from recent advances on fast static retrieval data structures with virtually no space overhead \cite{dillinger2021fast}.
In turn, our PHFs can be used to obtain improved \emph{updateable} retrieval data structures as discussed before.

\section{Preliminaries}
\label{sec:prelim}
\myparagraph{Cuckoo Hashing.}
Cuckoo hashing \cite{pagh2004cuckoo} is a well known approach to hash tables with open addressing.
In a basic cuckoo hash table, each object can be placed in one of two cells, determined by two hash functions.
Queries load the two candidate cells and compare both objects.
Insertion applies one of the hash functions and places the new object in the corresponding cell.
If the cell is already occupied, the object previously placed in that cell is pushed out and is recursively inserted using its other hash function.

Instead of locating each object in one of two cells, the idea can be generalized to $d$ cells \cite{fotakis2005space} by using $d$ hash functions.
In \emph{irregular} cuckoo hash tables, different objects can have a different number of choices \cite{dietzfelbinger2010tight}.
For example, some percentage of the objects get $d_1$ choices, some $d_2$ choices, and some $d_3$ choices.
Averaging over the $d_i$, the method enables $d$-ary cuckoo hashing with non-integer $d$ and higher load factors than a simple interpolation between two ordinary cuckoo hash tables \cite{dietzfelbinger2010tight}.

\myparagraph{Retrieval Data Structures.}
A \emph{retrieval data structure} or \emph{static function} on a set $S$ of objects describes
a function $f: S\rightarrow\{0,1\}^r$ that returns a specific $r$-bit value for each object.
Applying the function on an object not in $S$ can return an arbitrary value.
The lower bound of the space requirement of a retrieval data structure is $rN$ bits.
The best retrieval data structures now come very close to the lower bounds (around 1\% overhead) and are also quite fast \cite{dillinger2021fast}.

\myparagraph{PHF Construction by Cuckoo Hashing.}
To the best of our knowledge, constructing PHFs through cuckoo hashing was only mentioned very briefly before \cite{dillinger2021fast}.
In this paragraph, we give a more detailed and intuitive introduction to the idea.
A related idea is the construction of PHFs by solving a matching as described by, e.g., Botelho et al. \cite{botelho2013practical} and Navarro \cite[Section 4.5.3]{navarro2016compact}.

Assume that all input objects are inserted into a $d$-ary cuckoo hash table.
The table then implicitly describes an injective mapping from objects to table cells, because each cell only stores one object.
For perfect hashing, we are not interested in storing the objects themselves but only in mapping objects to numbers, e.g., table cell indices.
Because each object can only be placed in $d$ cells using $d$ hash functions $h_i$, we can remember the placement of each object by simply storing which of the hash functions was finally used to place the object.
We can do that using only about $\log d$ bits\footnote{Throughout this paper, $\log x$ stands for $\log_2 x$.} per object by constructing a retrieval data structure.
A query for an object $x$ then retrieves the hash function index $i(x)$ and executes $h_{i(x)}(x)$ to obtain a perfect hash function.

\myparagraph{Elias-Fano Coding.}
\label{sec:eliasFano}
Elias-Fano Coding \cite{Elias74, Fano71} is a way to efficiently store a monotonic sequence of $N$ integers.
It consists of two data structures, a bit vector $H$, and an array $L$.
An item at position $i$ is split into two parts.
The $\log N$ upper bits $u$ are stored as a 1-bit in $H[i + u]$.
The remaining lower bits are directly stored in $L$.
Items can be accessed in constant time by finding the $i$-th $1$-bit in $H$ using a $\textit{select}_1$ data structure and by looking up the lower bits in $L$.
The space usage of an Elias-Fano coded sequence is $2N+N\lceil \log U/N\rceil$ bits, where $U$ is the maximum value of an item.

\myparagraph{Golomb-Rice Coding.}
\label{sec:rice}
Golomb coding \cite{golomb1966run} with parameter $k$ can be used to store a sequence of integers that have a geometric distribution.
The idea is to store each integer $x$ as a quotient $q=\lfloor x/k \rfloor$ in unary coding and a remainder $x-qk$ in truncated binary coding.
Rice coding \cite{rice1979some} is Golomb coding where $k$ is a power of 2.
This makes arithmetics more efficient and simplifies storing the remainder to a normal array with binary coding.
Items can be accessed in constant time by looking up the array and reconstructing the quotient using a $\textit{select}_1$ query.

\section{Related Work}
\label{sec:related}
In the following, we first describe variants and enhancements of cuckoo hashing from the literature.
Afterwards, we introduce existing PHFs, most of which we later include in our experimental evaluation (see \cref{sec:experiments}).

\subsection{Cuckoo Hashing.}
After describing variants of cuckoo hashing, we describe construction algorithms and maximum load factors.

\myparagraph{Variants.}
Higher maximum load factors can be achieved by making the cells larger, so that they hold more than one object \cite{dietzfelbinger2007balanced}.
When then allowing the cells to overlap \cite{lehman20093}, even higher load factors are possible \cite{walzer2017load}.
For our application to perfect hashing, we only consider cells of size $1$.
On external memory, I/Os can be reduced by choosing candidate cells on the same page \cite{dietzfelbinger2011cuckoo}.
Maintaining two tables \cite{pagh2004cuckoo} of asymmetric size \cite{kutzelnigg2010improved} can improve the search time because more objects can be placed using their first hash function.
Giving each object $d$ instead of $2$ choices \cite{fotakis2005space} increases the maximum load factor.
\emph{Irregular} cuckoo hashing \cite{dietzfelbinger2010tight} uses a different number of hash function for each object.
A similar idea can also be found in coding theory, where each message bit is covered by an irregular number of check bits \cite{luby2001efficient}.
Specifically, the probability that a message bit is covered by $i$ check bits is proportional to $1/i$.
Another related result is the weighted Bloom filter \cite{bruck2006weighted}, where objects get a different number of hash functions (and therefore false positive probability) based on their query frequency and membership likelihood.

\myparagraph{Construction.}
The enhancement to $d$-ary cuckoo hashing \cite{fotakis2005space} makes insertions more complex because it is no longer clear which of the alternative cells to displace objects to.
Common ways to perform insertion are to find a shortest move sequence by performing \emph{breadth-first-search (BFS)} in a graph defining possible object moves
or by performing a \emph{random walk} in that graph. Both approaches need constant expected time when the table is not too highly loaded \cite{fotakis2005space,walzer2022insertion,frieze2009analysis,fountoulakis2013insertion,khosla2013balls, khosla2019faster}.

In this paper, we are interested in the static case, where all objects to be stored are known from the start.
In that case, it is also possible to construct the whole hash table at once instead of using incremental insertions.
Let us model the cuckoo hash table as a bipartite graph.
The first set of graph nodes is simply the set of input objects and the second set represents the table cells.
Edges connect each object to its candidate cells, as determined by the $d$ hash functions.
A matching of size $N$ then gives a collision free assignment from objects to table cells.
This can be calculated using, for example, the Hopcroft-Karp-Karzanov algorithm \cite{hopcroft1973n} or the LSA algorithm \cite{khosla2013balls, khosla2019faster}.

\myparagraph{Load Factors and Space Usage.}
\label{sec:loadFactors}
Classic cuckoo hashing with $d=2$ hash functions has a maximum load factor of at most $N/M=0.5$.
Using $d=4$ hash functions already increases the maximum load factor to $0.9768$ \cite{fountoulakis2012sharp, walzer2021peeling}.
In our construction, the load factor of the PHF equals the load factor of the cuckoo hash table, and the storage space is determined by the number of hash functions $d$.
This means that a PHF from binary cuckoo hashing with a load factor of $0.5$ can be represented using $\log 2=1$ bit per object.
A PHF with a load factor of $0.9768$ can be implemented using $2$ bits per object.

Ref. \cite{dietzfelbinger2010tight} gives maximum load factors for irregular cuckoo hashing, depending on the distribution of hash functions used.
When looking at a specific average number of hash functions $d'\in\mathds{R}$, the best load factors are given by combining objects with $\lfloor d' \rfloor$ and $\lceil d' \rceil$ hash functions \cite{dietzfelbinger2010tight}.
As we will see in \cref{sec:overloading}, this is not the case in the context of PHFs because we are looking at storage space instead of the average hash function.

\subsection{Perfect Hashing.}
\label{sec:relworkPerfectHashing}
The perfect hashing problem is already considered since the 1970s \cite{sprugnoli1977perfect, jaeschke1981reciprocal, brain1990perfect}, and is still an active area of research.
In the following paragraphs, we describe more recent papers.

\myparagraph{Order-Preserving.}
An order-preserving PHF maintains the order that the input objects are given in.
CHM \cite{czech1992optimal} and BMZ \cite{botelho2004new} construct an undirected graph with edges $\{(h_1(x),h_2(x)) | x \in S\}$.
They then assign a number to each vertex, such that for each edge, the sum of numbers stored in adjacent vertices gives the desired PHF value.
This can be done by assigning $0$ to an arbitrary vertex and then performing depth-first-search to assign all neighbors by simple subtraction.
CHM and BMZ store $2.09N$ and $1.15N$ integer numbers, respectively, therefore needing $\Oh{N \log(N)}$ space.
Note that space near $N\log N$ can also be achieved by explicitly storing the rank of the objects in a retrieval data structure.

\myparagraph{BDZ.}
In the BDZ algorithm \cite{botelho2013practical}, also called RAM algorithm or BPZ algorithm, each input object is mapped to an edge in a random hypergraph using $d$ independent hash functions $h_i$.
The hypergraph needs to be peelable\footnote{Possibility of obtaining a graph without edges by iteratively taking away edges that contain a node with degree 1 \cite{botelho2013practical,walzer2021peeling}.} in order to continue.
By peeling the graph, BDZ determines $i(x)$, such that $h_{i(x)}(x)$ is unique for each object $x$.
It then uses a linear equation system to determine a function $g$, such that
$i(x)=\left(\sum_{0 \leq i < d}g(h_i(x))\right) \textrm{ mod } d$.

Even though our presentation using cuckoo hashing sounds different, SicHash is similar to this idea.
The BDZ algorithm's task of finding a unique $h_{i(x)}(x)$ for each object can be interpreted as placing the objects in a cuckoo hash table.
The function $g$ serves as a retrieval data structure that maps each object to a hash function index $i(x)$.
The most important difference is that the BDZ algorithm couples retrieval and object placement by using the same set of hash functions.
In particular, $g$ is evaluated for the entire range $M$, so the space to store $g$ depends on $M$.
SicHash, in contrast, separates the two tasks of object placement and hash function retrieval.
This enables using a retrieval data structure of size $N$ instead.
Moreover, SicHash uses irregular cuckoo hashing, which cannot be represented efficiently with the integrated retrieval data structure of BDZ.
Finally, SicHash does not depend on peelability.

\myparagraph{WBPM.}
Weaver et al. \cite{weaver2020constructing} describe an algorithm for calculating MPHFs that is based on weighted bipartite matchings (WBPM).
The left set of the graph is determined by the $M=N$ input objects and the right set is determined by the $N$ possible hash values.
The edges are determined by applying $\Oh{\log(N)}$ hash functions to each object, where an edge determined from the $i$-th hash function has weight $i$.
The weighted matching can be solved with a weight of $1.83N$, giving an assignment from objects to hash values.
For storing which hash function to use for each object, WBPM uses a $1$-bit retrieval data structure.
The keys to the retrieval data structure are tuples of object and hash function index.
The stored value is $1$ for the hash function to finally be used, and $0$ for all smaller indices.
The weight of $1.83N$ therefore also equals the space usage of the final data structure, except for overheads like prefix sums due to bucketing.

SicHash uses a similar structure but simplifies each of the ingredients.
Instead of a \emph{weighted} bipartite matching, SicHash (implicitly) solves a \emph{non-weighted} bipartite matching by constructing a cuckoo hash table.
Instead of querying a $1$-bit retrieval data structure multiple times for each hash function evaluation, SicHash only performs a single query to a retrieval data structure.
While WBPM constructs a retrieval data structure consisting of $1.83N$ objects, SicHash generates retrieval data structures with a total of $N$ objects, which makes the construction faster.
While WBPM's space usage is competitive for MPHFs, constructing a non-minimal PHF is less efficient.
With a load factor of $0.85$, for example, SicHash achieves a space usage of $1.43N$ bits, while our preliminary experiments show that a matching like above has a weight of $1.54N$.
This stems from the fact that SicHash stores the selected hash function index using binary code, while WBPM effectively uses unary code.

\begin{figure*}[t]
  \input{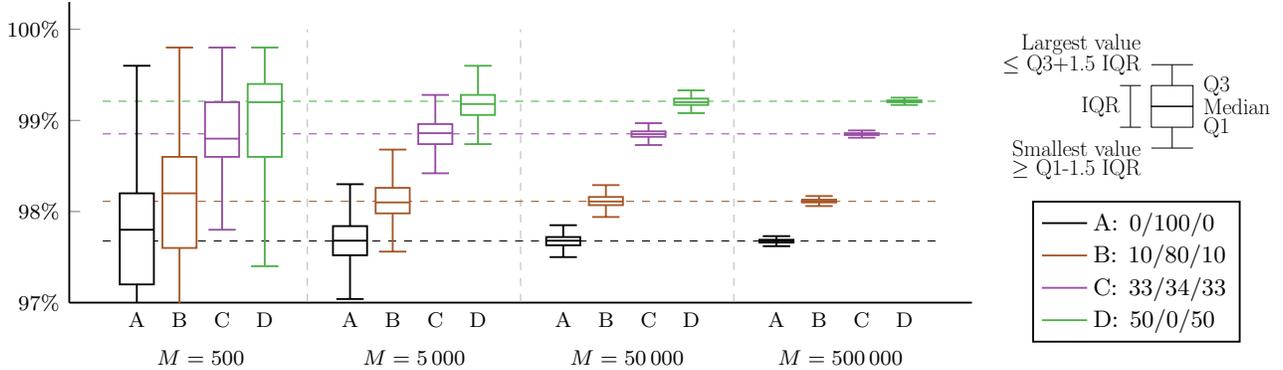}
  \caption{Achieved load factors when running different irregular cuckoo hashing configurations, which all need the same storage space (2 bits). The configurations are described by the percentages of objects with 2/4/8 choices, having a space usage of 1/2/3 bits, respectively. The configuration 0/100/0 refers to ordinary 4-ary cuckoo hashing. Horizontal lines indicate numerically calculated load thresholds (see \cref{sec:loadThresholds}).}
  \label{fig:loadFactorBoxplot}
\end{figure*}

\myparagraph{CHD \cite{belazzougui2009hash} / FCH \cite{fox1992faster} / PTHash \cite{pibiri2021pthash}.}
\label{sec:pthash}
The basic idea of this approach is to hash the input objects to a number of buckets with expected equal size.
After sorting the buckets by their size, the methods search for a hash function that stores the entire bucket in the final output domain without causing collisions.
CHD \cite{belazzougui2009hash} tries out hash functions linearly, so that the data to store for each bucket can be compressed efficiently.
With a load factor of $81\%$, CHD can construct a PHF using $1.4$ bits per object.
With a load factor of $99\%$, it achieves $1.98$ bits per object.
In FCH \cite{fox1992faster}, the bucket sizes are asymmetric --- using the default parameters, 60\% of the objects are mapped to 30\% of the buckets.
After hashing the objects of a bucket, it searches for a rotation value that (mod $M$) places all objects into the output range without collisions.
FCH produces MPHFs with about 2 bits per object.
PTHash \cite{pibiri2021pthash} combines the two ideas by using asymmetric buckets and trying hash functions linearly to enable a compressible representation.
To speed up searching for the hash function in each bucket, it first generates a non-minimal PHF.
Instead of using the well known \emph{rank trick} \cite{botelho2007simple} to convert a PHF to an MPHF, PTHash uses a new approach that enables faster queries:
We can interpret the values in the output range $M$ that no object is mapped to as \emph{holes}.
PTHash now stores an array of size $M-N$ that re-maps each object with a hash value $> N$ to one of the holes.
Given that the holes are a monotonic sequence of integers, the array can be compressed with Elias-Fano coding.
Because most objects are mapped to values $\leq N$, the expensive lookups into the Elias-Fano sequence are only rarely needed.

\myparagraph{MeraculousHash \cite{CSLSR11} / FiPHa \cite{MSSZ14} / BBHash \cite{limasset2017fast}.}
This approach hashes the objects to $\gamma N$ buckets, for $\gamma \in \mathds{R}, \gamma \geq 1$.
If a bucket has exactly one object mapped to it, that mapping is injective.
Objects from buckets that hold more than one object are bumped to the next layer of the same data structure.
For each layer, a bit vector of size $\gamma N$ indicates which buckets had stored a single object, e.g., where the recursive query can stop.
Together with a rank data structure, that bit vector can also be used to make the resulting hash function minimal.
This approach allows fast construction and queries but needs space at least $e$ bits per object \cite{MSSZ14} -- more than $4$ for really good speed. 

\myparagraph{RecSplit.}
RecSplit \cite{esposito2020recsplit} distributes all objects to buckets of expected equal size.
Within each bucket, RecSplit first searches for a hash function with binary output that splits the set of objects in half.%
\footnote{The number of splits is larger for the bottom recursion layers.}
This is repeated recursively until the set of objects has a small, configurable size.
Usual values for the leaf size are about 8--12 objects.
At a leaf, RecSplit then tries out hash functions until it finds one that is a bijection.
For each bucket, RecSplit stores the hash function seeds that make up the splitting tree, the seeds for each leaf, and a prefix sum of the bucket sizes.
There are configurations that need only $1.56$ bits per object.

\section{Overloading}
\label{sec:overloading}
In cuckoo hashing, the maximum load factor of a table is a widely studied subject (see \cref{sec:loadFactors}).
For example, it is well known that a table with $d=2$ hash functions has a \emph{load threshold} of $50\%$.
For $M\rightarrow\infty$, the probability of successful table construction with load factor $>50\%$ approaches $0$, while it approaches $1$ for load factors $<50\%$.
Let us now look at a very small table of size $M=3$ storing $N=2$ objects.
This table has a load factor of $\approx 66\%$, which is more than load the threshold.
Still, the probability of successful construction, i.e., not all four hash function values being the same, is $1-3\cdot(1/3)^4 \approx 88\%$.
This shows that the load factors can be considerably higher than the asymptotic limits when using small tables.
We call a table that contains more objects than the asymptotic limit \emph{overloaded}.

To illustrate the achievable load actors, we incrementally construct cuckoo hash tables.
\Cref{fig:loadFactorBoxplot} gives a box plot for the load factors at which the insertion finally fails.
It shows three fundamental observations that we use in SicHash to increase the load factor while decreasing the amount of memory needed.

\myparagraph{(1) Variance.}
Unsurprisingly, small tables have a higher variance in their achieved load factors.
Therefore, improved load by the standard deviation is possible by just retrying a constant number of times in expectation.

\begin{figure*}[t]
  \centering
  \includegraphics[scale=0.8]{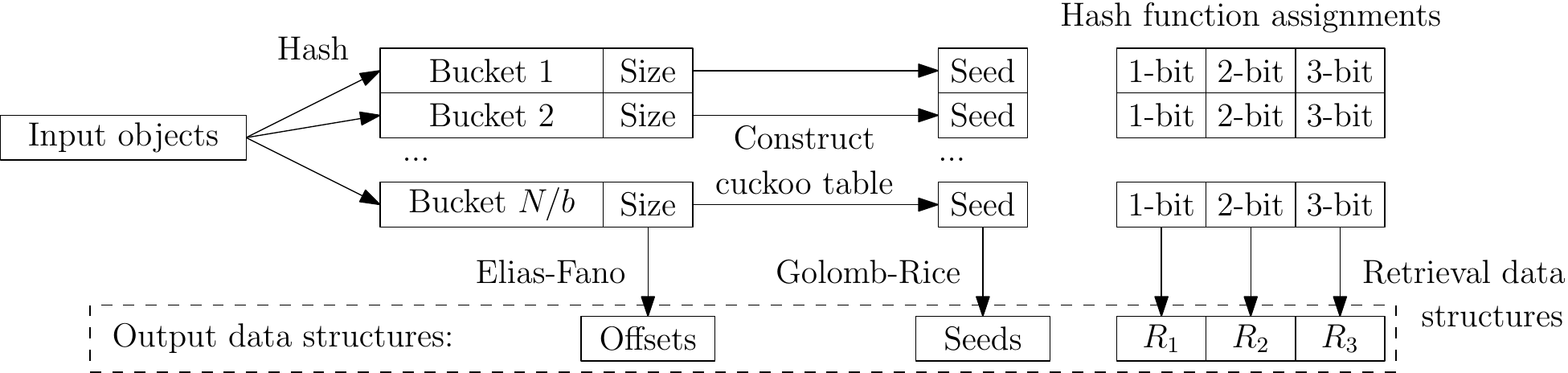}
  \caption{Data flow during construction. Objects are hashed to buckets of expected equal size. Within each bucket, a cuckoo hash table is constructed. The resulting hash function assignments from all small hash tables are stored together in three large retrieval data structures.}
  \label{fig:constructionDataFlow}
\end{figure*}

\myparagraph{(2) Median.}
For some configurations, small tables not only enable higher load factors because of the variance, but also because their median is higher than the load threshold.
The effect is even more pronounced for ordinary binary cuckoo hashing, where a table with $M=500$ has a median load factor of 56\% (see \cref{sec:overloadBinary}).
A similar observation can be found in perfect hashing:
The information theoretic minimal space to store an MPHF with small $N$ is significantly lower than the asymptotic value for $N\rightarrow\infty$ \cite{weaver2020constructing}.

\myparagraph{(3) Space Usage.}
A metric for the lookup efficiency in irregular cuckoo hash tables is the average number of hash functions.
For a specific average number of hash functions $d'\in\mathds{R}$, the best load factors are given by combining objects with $\lfloor d' \rfloor$ and $\lceil d' \rceil$ hash functions \cite{dietzfelbinger2010tight}.
This picture changes fundamentally in the context of PHFs because the choice of hash functions is stored in binary coding.
While, for example, an irregular cuckoo hash table with $50\%$ $2$-choice and $50\%$ $4$-choice has $3$ choices on average, the storage space of the corresponding PHF is only $0.5 \log(2) + 0.5 \log(4) = 1.5 < \log(3)$.
The configurations in \cref{fig:loadFactorBoxplot} all need the same storage space but the median load factor increases the farther we are from the optimal configuration derived in Ref. \cite{dietzfelbinger2010tight}.

\myparagraph{Conclusion.}
Smaller cuckoo hash tables give better load factors than larger ones.
Equivalently, by making the tables smaller, we can achieve the same load factor using a hash function mixture that needs less space.
Even though we have to store a seed because the variance in the load factors is higher, we can use the effect of overloading to save overall space.

\section{SicHash Perfect Hash Functions}
\label{sec:theDataStructure}
The following section introduces the main result of our paper: SicHash perfect hash functions.
SicHash combines four results: Using cuckoo hashing for perfect hashing, making cuckoo hashing irregular, overloading small tables, and fast space efficient retrieval.

\subsection{Construction.}
Building a SicHash function consists of the main steps partitioning, cuckoo hashing, storing bucket metadata, and constructing the retrieval data structures.
\Cref{fig:constructionDataFlow} gives an overview over the process.

\myparagraph{Partitioning.}
First, we hash the input objects to a number of buckets that all have the same expected size $b$, for example $b = 5000$.
The small size enables overloading and also keeps the storage space during construction small enough that the whole table fits into the cache.

\myparagraph{Cuckoo Hashing.}
Within each bucket, we generate an irregular cuckoo hash table, using the same load factor as the overall perfect hash function.
The number of cells in each of the small cuckoo hash tables is determined by the number of objects hashed to it, so the small tables have different sizes.
However, the probability of a successful construction is similar for all of them.
To determine how many hash functions (and therefore candidate cells) should be used for each object, we hash each object to a \emph{class}.
A certain percentage $p_1$ of objects is placed with $d=2$ choices, a percentage $p_2$ with $d=4$ choices, and the remaining objects with $d=8$ choices.

To insert objects into the hash tables, we use \emph{rattle kicking} \cite{kuszmaul2016kickout} instead of the classical random walk.
Rattle kicking maintains a counter for each object, indicating how often it was moved to a new cell.
An object is then only evicted from its cell when its rattle counter is lower than the counter of the object to be inserted into the same cell.
The hash function index to use for inserting is the rattle counter modulo $d$.
For normal cuckoo hash tables, storing the counter would decrease the space efficiency.
In our case, we only store the hash table temporarily, and also need the hash function index to construct the PHF later.
This makes SicHash an attractive application for rattle kicking.
With rattle kicking, we can avoid the cost of random number generation and empirically need a smaller number of steps during the insertion.
Constructing a bucket may fail, in particular, when we configure a high degree of overloading.
In this case, construction is retried while incrementing a seed value determining the used hash functions.

\begin{figure*}[t]
  \centering
  \includegraphics[scale=0.8]{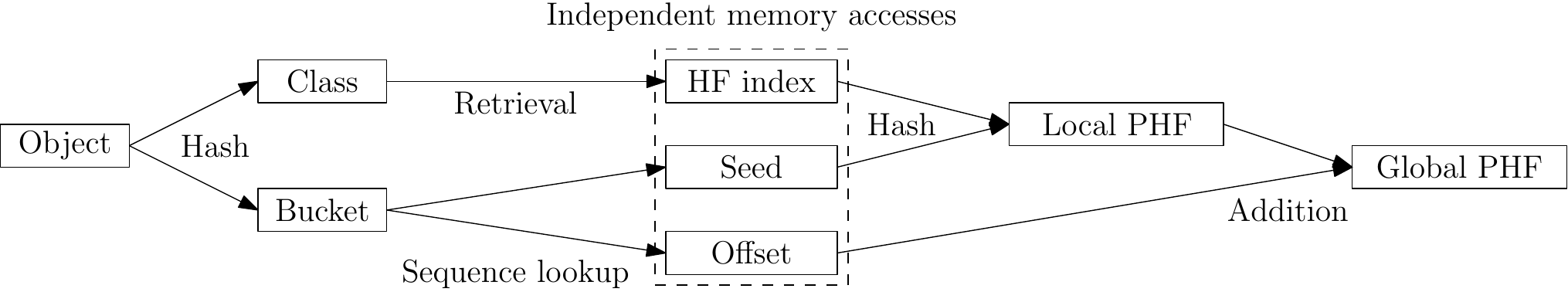}
  \caption{Data flow during a query. From the class of an object, we get the retrieval data structure to query. From the bucket index, we get the bucket's seed and offset. Combining this, we get a global hash function value.}
  \label{fig:queryDataFlow}
\end{figure*}

\myparagraph{Storing Bucket Metadata.}
The result is a number of small hash tables, each with a seed leading to successful construction.
In order to determine a global PHF, we need to offset each small table.
We can do that by storing the exclusive prefix sum of table sizes using Elias-Fano coding.
When trying out seeds, we can simply count up, starting with $0$.
Construction usually succeeds with one of the first few seeds, so we can compress the list of seeds using Golomb-Rice Codes.
In practice, storing the two sequences in arrays offers better query time and negligible space overhead.

\myparagraph{Retrieval.}
Now we only need to store the assignment from objects to cells within the small hash tables by storing which of the hash functions finally placed each object.
Because the hash tables are irregular, we get indices of $1$, $2$, and $3$ bits.
While small hash tables enable overloading, retrieval data structures in contrast profit from handling many objects and can achieve overheads as low as $1\%$ \cite{dillinger2021fast}.
We therefore build $3$ large retrieval data structures that hold the $1$, $2$, and $3$-bit values from \emph{all} the small hash tables.
The space usage of the final PHF is dominated by the retrieval data structures.

\myparagraph{Parameters.}
A SicHash PHF has three main tuning parameters: The load factor $\alpha=N/M$ to try construction for, and the class sizes for irregular cuckoo hashing, $p_1$ and $p_2$.
Ignoring overloading, it is then possible to numerically determine a configuration that maximizes the load factor (see \cref{sec:loadThresholds}).
Numerically calculating efficient configurations \emph{with} overloading remains an open problem.

When a specific space budget $\beta$ for the retrieval data structures is given, we get $p_2=3-\beta-2p_1$.
By transforming $0 \leq p_1,p_2,p_3 \leq 1$, we then get a range of valid choices for $p_1$ that we can interpolate between.
An equivalent and more natural set of tuning parameters therefore is $(\alpha, \beta, x)$, where $x$ interpolates between different configurations that need the same amount of space.
The thresholds can then be calculated as follows.

\begin{align*}
p_{1,\textrm{min}} = max(0, 2 - \beta), \qquad
          p_{1,\textrm{max}} = (3 - \beta) / 2\\
p_1 = p_{1,\textrm{min}} + x(p_{1,\textrm{max}} - p_{1,\textrm{min}}), \qquad
          p_2 = 3 - 2 p_1 - \beta
\end{align*}

\subsection{Queries.}
A query for an object first hashes it to find its class, e.g., its number of candidate cells in the small hash table.
This determines which retrieval data structure needs to be queried for the hash function index.
Additionally, the object is hashed to find its bucket, and therefore the seed and offset.
The value of the PHF is then given by hashing the object with the retrieved hash function index and the bucket's seed, and by adding the bucket's offset.
\cref{fig:queryDataFlow} gives an illustration for the data flow during a query, showing that all three memory access operations are independent of each other, so they can be prefetched or performed in parallel.

\section{Enhancements}
\label{sec:enhancements}
SicHash lends itself to numerous enhancements, which we outline below.

\myparagraph{Minimal Perfect Hashing.}
\label{sec:minimal}
SicHash can be converted to an MPHF by applying the same technique as in PTHash \cite{pibiri2021pthash}.
The idea is to re-map objects with hash function values $>N$ to smaller values by using an Elias-Fano coded sequence of size $M-N$.
Refer to \cref{sec:pthash} for details.
Note that the number of lower bits of the Elias-Fano coded sequence only depends on the load factor that was used before re-mapping.
In practice, we can therefore use a compile-time parameter for faster bit operations in the Elias-Fano sequence.

\myparagraph{Parallel Construction.}
Perfect hash functions can always be parallelized trivially by introducing a new layer on top of the data structure.
SicHash can be parallelized more directly and without effect on the query speed.
The small cuckoo hash tables of each bucket can be constructed embarrassingly parallel.
Retrieval data structures can also be computed in parallel at some small space overhead linear in the number of processors and without query overhead \cite{dillinger2021fast}.

\myparagraph{External Memory Construction.}
SicHash can be adapted to very large inputs: First use external sorting to partition the objects into buckets.
Then, for each bucket, find a cuckoo hash table --- outputting sequences of seeds, offsets, and key-value pairs for the retrieval data structures.
The latter can be fed into an external memory construction of the retrieval data structures \cite{dillinger2021fast}.%
\footnote{BuRR \cite{dillinger2021fast} sorts elements by a hashed starting position in an equation system. By making this position monotonic in the bucket index one could save that sorting step.}

\section{Analysis}
\label{sec:analysis}

\begin{theorem}
Parameters of SicHash can be chosen such that the expected construction time is linear in the input size $N$, that the query time is constant, and that the expected space consumption is $\Oh{N}$ bits.
\end{theorem}
\begin{proof}
  (Outline) To be able to show linear construction time, we look at the case of
constant expected bucket size, a bucket load factor that
ensures constant success probability of cuckoo hash table construction, space efficient encoding of
offsets using Elias-Fano coding, Golomb-Rice coding of seeds, and cuckoo table insertion using BFS.
We also assume that we  use a retrieval data structure that needs $\Oh{N}$ bits, can be constructed in linear time and supports queries in constant time.
After this proof, we informally discuss what changes for the simple implementation used in our experiments.

\myparagraph{Query time.}
A query evaluates a constant number of hash functions (constant time), performs one access to a retrieval data structure (constant time as assumed above), and decodes a number each in an Elias-Fano coded and a Golomb-Rice coded sequence.
Decoding the numbers boils down to constant time $\textit{select}_1$-operations in a bit vector (see \cref{sec:eliasFano}).

\myparagraph{Space.}
The retrieval data structure uses a linear number of bits by our assumption.
With average bucket size $b$, the Elias-Fano data structure takes $\left(2+\log\frac{Mb}{N}\right)\frac{N}{b}=\Oh{N}$ bits.
For constant success probability of construction, the seed has a geometric distribution and constant expected length, i.e.,  expected space consumption $\Oh{\frac{N}{b}}=\Oh{N}$ bits.
When constructing an MPHF, the additional Elias-Fano coded sequence (see \cref{sec:minimal}) takes $(M-N)(2+\log\frac{N}{M-N}) = N\frac{1}{\alpha-1}(2 + \log\frac{\alpha}{1-\alpha})=\Oh{N}$ bits, where $\alpha$ is the load factor before compaction.

\myparagraph{Construction Time.}
Constructing the retrieval data structures takes linear time by assumption.
Building the data structures for seeds and offsets is obviously possible in linear time.
Construction time for a bucket is at most quadratic in the bucket size (and, with constant success probability, retries contribute only a constant factor in expectation).
With constant expected bucket size, we get the same execution time as a bucket sorting algorithm that uses a quadratic algorithm per bucket, which is expected linear \cite[Theorem~5.9]{SMDD19}.
A similar argument can also be found in RecSplit \cite{esposito2020recsplit}.
\end{proof}

Our implementation gains simplicity and query speed by storing offsets and seeds directly using $\log M$ bits.
For this to be space efficient, we would need average bucket size $b=\Om{\log N}$.
At least with our simple estimation of construction time, this would result in superlinear construction time of $\Om{N\log N}$.
This can be improved using faster construction algorithms. For example, using the Hopcroft-Karp-Karzanov \cite{hopcroft1973n} algorithm, time $\Oh{N\sqrt{\log N}}$ can be proven.
Assuming a result for random graph matching \cite{bast2006matching} also transfers to our case, we would even get $\Oh{N\log\log N}$.
At least when we are sufficiently far from the load threshold, various previous results indicate that linear construction time is possible \cite{fotakis2005space,walzer2022insertion,frieze2009analysis,fountoulakis2013insertion,khosla2013balls, khosla2019faster}. With overloading, tight construction time bounds remain an open problem.

\myparagraph{Load Thresholds.}
\label{sec:loadThresholds}
For a fixed configuration of fractions $p₁$, $p₂$, $p₃$, we now use existing theory to derive a threshold for the asymptotically maximum load factor $c^* = c^*(p₁,p₂,p₃)$ when the bucket size is large.
The hope is that the threshold is also indicative of achievable loads for small to medium bucket sizes.
We use a framework by Lelarge \cite{L:A_New_Approach:2012}, which is quite challenging to understand and formally apply.
Introducing the underlying ideas of belief propagation and local weak convergence, let alone carrying out the required technical work for our case is beyond the scope of this paper.
What we can do is restate the equations characterizing the load threshold $c_d^*$ when each object has the same number $d$ of hash functions and then point out how these equations have to be adapted for the irregular case.
The following equations are taken from Ref. \cite[Theorems 1.1 and 4.1]{L:A_New_Approach:2012}, specialised for $(h,k,\ell) = (d,1,1)$ and simplified.

\begin{align*}
    F(q,c) &= 1-(1-g^B(q))^d + \tfrac{1-\e^{-cdq}(1+cdq)}{c}\\
    g^B(q) &= \e^{-cdq}, \qquad g^A(p) = (1-p)^{d-1}\\
    X(c) &= \{F(q,c) \mid q ∈[0,1], g^A(g^B(q)) = q \}\\
    c_d^* &= \sup\{c > 0 \mid \inf X(c) = 1\}.
\end{align*}

Modified for our case, when objects use $d₁$=2, $d₂$=4, $d₃$=8 hash functions with probability $p₁$, $p₂$, $p₃$, respectively, not much changes in the terms that handle belief propagation at the table cells, except that $d$ has to be replaced with $\bar{d} := \sum_i p_i d_i$.
From the point of view of an object, three possibilities have to be taken into account, which occur with probabilities $p₁$, $p₂$, $p₃$ if the object is chosen uniformly at random, and with probabilities $p₁d₁/\bar{d}$, $p₂d₂/\bar{d}$, $p₃d₃/\bar{d}$ if the object is chosen with a probability proportional to the number of table cells it is incident to.
The correspondingly modified equations read as follows.

\begin{align*}
    F(q,c) &= 1-\sum_i p_i (1-g^B(q))^{d_i} + \tfrac{1-\e^{-c\bar{d}q}(1+c\bar{d}q)}{c}\\
    g^B(q) &= \e^{-c\bar{d}q}, \qquad g^A(p) = \sum_i \frac{p_i d_i}{\bar{d}} (1-p)^{d_i-1}\\
    X(c) &= \{ F(q,c) \mid q ∈[0,1], g^A(g^B(q)) = q\}\\
    c^* &= c^*(p₁,p₂,p₃) = \sup\{c > 0 \mid \inf X(c) = 1\}.
\end{align*}

Consider solutions $(q,c)$ to $g^A(g^B(q)) = q$.
For any $c$, a trivial solution is $(0,c)$.
Any non-trivial solution is uniquely determined by $λ := c\bar{d}q$ since we can compute $q(λ) = g^A(g^B(q)) = g^A(e^{-λ})$ and $c(λ) = \frac{λ}{q(λ)\bar{d}}$ from it.
The corresponding value $F(λ) = F(q(λ),c(λ))$ is

\[
    F(λ) = 1-\sum_i p_i (1-\e^{-λ})^{d_i} + \frac{g^A(e^{-λ})\bar{d}}{λ}(1-\e^{-λ}(1+λ))
\]

We can then rewrite $c^* = \sup\{c > 0 \mid ∄λ > 0: F(λ) < 1 \text{ and } c = c(λ) \}$
and obtain numerical approximations of $c^*$ by plotting $F(λ)$ and $c(λ)$. Indeed, it seems that $c^* = c(λ^*)$ for the largest root $λ^*$ of $F(λ)-1$.
Needless to say, a lot more work would be required to rigorously defend this method.

\section{Experiments}
\label{sec:experiments}
The code and scripts needed to reproduce our experiments are available on GitHub under the General Public License:
\url{https://github.com/ByteHamster/SicHash}.
The repository also contains a Docker container that can build and run a simplified version of the experiments in about 90 minutes.

\definecolor{veryLightGrey}{HTML}{F2F2F2}
\definecolor{colorBmz}{HTML}{000000}
\definecolor{colorBdz}{HTML}{E41A1C}
\definecolor{colorFch}{HTML}{444444}
\definecolor{colorChd}{HTML}{377EB8}
\definecolor{colorChm}{HTML}{A65628}
\definecolor{colorHeterogeneous}{HTML}{4DAF4A}
\definecolor{colorPthash}{HTML}{984EA3}
\definecolor{colorRecSplit}{HTML}{FF7F00}
\definecolor{colorBbhash}{HTML}{F781BF}

\myparagraph{Experimental Setup.}
We run our experiments on an Intel Xeon E5-2670 v3 with a base clock speed of 2.3 GHz.
The machine runs Ubuntu 20.04 with Linux 5.10.0.
We use the GNU C++ compiler version 11.1.0 with optimization flags \texttt{-O3 -march=native}.
We also achieve comparable results on an AMD Ryzen 3950X with a base clock speed of 3.5 GHz.
After small adaptions of the competitors to replace inline assembly and intrinsics, we even achieve comparable results on an ARM Neoverse-N1.
Refer to \cref{sec:additionalExperiments} for measurements on the two additional machines.
As a retrieval data structure, we use Bumped Ribbon Retrieval \cite{dillinger2021fast}, with two alternative configurations: $w=64$ with 2-bit bumping info, and $w=32$ with 1-bit bumping info.
Instead of encoding the per-bucket metadata using Elias-Fano coding and Golomb-Rice coding, we use a plain array, which is faster to access and causes only a small space overhead for sufficiently large buckets.
The extension to minimal perfect hashing stores the re-mapping with Elias-Fano coding, which is based on sdsl's \cite{gog2014theory} arrays of flexible bit width and the select data structures by Kurpicz \cite{kurpicz2022pasta}.
For all our experiments, construction and queries are executed on a single thread.
Objects are strings of uniform random length $\in [10, 50]$ containing random characters except for the zero byte.

\subsection{SicHash Configurations.}
The query performance is independent of the choice of $p_1$ and $p_2$, so we focus the comparison on the construction performance.

\myparagraph{Cuckoo Placement.}
It is possible to construct the cuckoo hash table with random walk insertion, as well as matching based methods.
In our experiments, the random walk variant \emph{rattle kicking} \cite{kuszmaul2016kickout} is usually faster than a construction based on Hopcroft-Karp-Karzanov \cite{hopcroft1973n}.

\begin{figure}[t]
        \definecolor{veryLightGrey}{HTML}{DDDDDD}
    \definecolor{greenToPurple1}{HTML}{762A83}
    \definecolor{greenToPurple2}{HTML}{AF8DC3}
    \definecolor{greenToPurple3}{HTML}{E7D4E8}
    \definecolor{greenToPurple4}{HTML}{D9F0D3}
    \definecolor{greenToPurple5}{HTML}{7FBF7B}
    \definecolor{greenToPurple6}{HTML}{1B7837}
    \pgfplotscreateplotcyclelist{bucketSizeList}{%
        greenToPurple1,mark=o\\%
        greenToPurple2,mark=triangle\\%
        greenToPurple3,mark=pentagon\\%
        greenToPurple4,mark=square\\%
        greenToPurple5,mark=x\\%
        greenToPurple6,mark=diamond\\%
    }
    \centering
    \begin{tikzpicture}
        \begin{axis}[
            xlabel={Bits per object},
            ylabel={MObjects/second},
            width=3.5cm,
            height=3.5cm,
            cycle list name=bucketSizeList,
            mark options={mark indices=1},
            every axis plot/.append style={thick},
          ]
          \addplot coordinates { (2.89058,4.04858) (2.7993,4.01606) (2.76861,3.97456) (2.6778,3.92773) (2.64802,3.885) (2.64724,3.83142) (2.61739,3.79651) (2.61736,3.67107) (2.58679,3.58423) (2.55692,3.09981) (2.53621,3.01568) (2.51627,2.98507) (2.49666,2.91206) (2.49608,2.80269) (2.49572,2.67523) (2.476,2.57599) (2.46562,2.55885) (2.45515,2.28624) (2.43559,2.10172) (2.43545,2.09996) (2.43498,1.84638) (2.41557,1.77368) (2.41548,1.72951) (2.41537,1.63773) (2.40598,1.63132) (2.4055,1.56397) (2.39494,1.43225) (2.37549,1.2213) (2.37533,1.21065) (2.37532,1.15848) (2.37491,1.07319) (2.37441,1.03648) (2.37412,0.972763) (2.35551,0.961723) (2.35455,0.892857) (2.34454,0.828912) (2.33466,0.737898) (2.33459,0.700476) (2.3337,0.647585) (2.31478,0.602192) (2.31442,0.55054) (2.31414,0.517866) (2.29404,0.459982) (2.29352,0.416945) (2.28457,0.4) (2.28438,0.395351) (2.27397,0.34118) (2.25424,0.248731) (2.25355,0.248373) (2.25342,0.218274) (2.23322,0.174862) (2.21319,0.124085) (2.21295,0.108899) (2.19324,0.0834112) (2.19303,0.0711572) (2.19271,0.0649528) (2.17261,0.0549282) (2.16274,0.0448684) (2.13282,0.0226486) (2.13274,0.0163669) (2.10175,0.0102935) };
          \addlegendentry{$b=100$};
          \addplot coordinates { (2.33927,4.04531) (2.31879,3.96197) (2.23733,3.95883) (2.2176,3.92157) (2.17653,3.92157) (2.13711,3.82263) (2.13709,3.72024) (2.11655,3.61795) (2.09626,3.53607) (2.05551,3.42466) (2.05535,3.23415) (2.03576,3.1407) (2.01524,2.88517) (2.01518,2.69107) (1.99562,2.405) (1.9951,2.34192) (1.99491,2.17391) (1.97494,1.91939) (1.97466,1.79856) (1.95468,1.65948) (1.95466,1.52114) (1.95433,1.24502) (1.93456,1.19646) (1.9341,1.11309) (1.91476,0.937383) (1.91422,0.820749) (1.91368,0.73014) (1.89425,0.647501) (1.8937,0.585206) (1.87399,0.429295) (1.87362,0.314584) (1.85382,0.275361) (1.85376,0.234434) (1.85334,0.194311) (1.83342,0.1631) (1.83333,0.132233) (1.81319,0.0871384) (1.81292,0.0724386) (1.79342,0.0458278) (1.79289,0.0338327) (1.77314,0.0214473) };
          \addlegendentry{$b=200$};
          \addplot coordinates { (2.49245,4.37828) (2.2804,4.37063) (2.251,4.363) (2.15942,4.3554) (2.12892,4.329) (2.06825,4.28816) (2.03819,4.28082) (2.03802,4.26257) (1.97757,4.20521) (1.97704,4.17711) (1.94705,4.17014) (1.85544,4.14938) (1.82435,4.02253) (1.81565,3.89712) (1.79555,3.81971) (1.77547,3.73413) (1.76459,3.61272) (1.75473,3.48918) (1.73528,3.1407) (1.73477,3.0175) (1.73456,2.61097) (1.71444,2.59336) (1.71438,2.1796) (1.70404,2.17014) (1.69503,1.97628) (1.69439,1.49611) (1.67429,1.36986) (1.67396,1.0408) (1.65364,0.828638) (1.64411,0.630915) (1.64363,0.581869) (1.63356,0.437063) (1.61388,0.202421) (1.61338,0.198342) (1.61323,0.10062) (1.59362,0.0731764) (1.59339,0.0550673) (1.59294,0.0319538) (1.57325,0.0228444) };
          \addlegendentry{$b=500$};
          \addplot coordinates { (2.34895,4.54133) (2.31861,4.5208) (2.16711,4.4964) (2.10608,4.43656) (2.07676,4.43262) (2.01461,4.42478) (1.98402,4.41306) (1.9549,4.363) (1.9239,4.3554) (1.86394,4.329) (1.8639,4.27716) (1.83354,4.2517) (1.83332,4.22297) (1.79312,4.2123) (1.71236,4.1632) (1.71208,4.11184) (1.71204,4.095) (1.71188,4.08163) (1.69198,4.05515) (1.65173,3.93082) (1.65144,3.9032) (1.63134,3.87597) (1.62188,3.87597) (1.62097,3.73134) (1.6103,3.65764) (1.5908,3.63636) (1.59077,3.51124) (1.57077,3.39674) (1.56048,3.00661) (1.5503,2.74123) (1.55011,2.57865) (1.52996,1.77873) (1.51073,0.381563) (1.49956,0.104106) (1.49941,0.101096) (1.48967,0.00597334) };
          \addlegendentry{$b=5000$};
          \addplot coordinates { (2.02332,3.8432) (1.98296,3.8373) (1.90141,3.79939) (1.84174,3.74251) (1.80073,3.71195) (1.78145,3.7037) (1.78085,3.64166) (1.76055,3.62845) (1.74017,3.56125) (1.72133,3.4626) (1.70065,3.4626) (1.68005,3.42231) (1.63967,3.39674) (1.61997,3.26158) (1.59922,3.13873) (1.59858,2.96209) (1.57886,2.91886) (1.55847,2.76549) (1.55824,2.49252) (1.53819,2.34412) (1.51817,1.50466) (1.49847,0.0429852) };
          \addlegendentry{$b=20000$};
          \addplot coordinates { (2.30314,3.69549) (2.27434,3.67107) (2.18167,3.663) (2.12152,3.64166) (2.09162,3.56633) (1.9997,3.54108) (1.94028,3.5137) (1.93962,3.46981) (1.84814,3.465) (1.81773,3.4626) (1.78751,3.37154) (1.75712,3.32447) (1.72715,3.2113) (1.69699,3.2113) (1.66647,3.04507) (1.63686,2.88517) (1.61541,2.8393) (1.60623,2.72926) (1.59624,2.67094) (1.57605,2.53936) (1.57585,2.28415) (1.57532,2.28415) (1.57507,2.21435) (1.55574,2.13038) (1.5557,2.09468) (1.54567,2.08247) (1.53522,1.85667) (1.53484,1.55473) (1.51542,1.12133) (1.51533,1.07181) (1.51453,0.156357) };
          \addlegendentry{$b=100000$};

          \legend{};
        \end{axis}
    \end{tikzpicture}
    \hfill
    \begin{tikzpicture}
        \begin{axis}[
            xlabel={Bits per object},
            legend columns=3,
            legend to name=legendParetoPlotsByBucketSize,
            width=3.5cm,
            height=3.5cm,
            cycle list name=bucketSizeList,
            mark options={mark indices=1},
            every axis plot/.append style={thick},
          ]
          \addplot coordinates { (2.09929,4.01606) (2.06861,3.97456) (1.97781,3.92773) (1.94804,3.885) (1.94728,3.83142) (1.91744,3.79651) (1.88694,3.58423) (1.85742,3.09981) (1.83676,3.01568) (1.81684,2.98507) (1.79733,2.91206) (1.79685,2.80269) (1.79657,2.67523) (1.77699,2.57599) (1.76669,2.55885) (1.75656,2.28624) (1.73736,2.10172) (1.73726,2.09996) (1.73686,1.84638) (1.71808,1.77368) (1.7089,1.63132) (1.70862,1.56397) (1.69858,1.43225) (1.69856,1.38619) (1.6802,1.2213) (1.68012,1.21065) (1.67995,1.03648) (1.66198,0.961723) (1.66154,0.892857) (1.65228,0.828912) (1.64368,0.737898) (1.62944,0.602192) (1.61088,0.459982) (1.60278,0.4) (1.60272,0.395351) (1.59406,0.34118) (1.5801,0.248731) (1.57948,0.248373) (1.56277,0.174862) (1.5491,0.124085) (1.53336,0.0834112) (1.51969,0.0549282) (1.51201,0.0448684) (1.49206,0.0226486) (1.4724,0.0102935) };
          \addlegendentry{$b=100$};
          \addplot coordinates { (1.99431,4.04531) (1.97384,3.96197) (1.89236,3.95883) (1.87264,3.92157) (1.83157,3.92157) (1.79216,3.82263) (1.77169,3.61795) (1.75143,3.53607) (1.71073,3.42466) (1.71065,3.23415) (1.69113,3.1407) (1.67082,2.88517) (1.65165,2.405) (1.65122,2.34192) (1.65084,2.17391) (1.6314,1.91939) (1.63127,1.79856) (1.61211,1.65948) (1.59339,1.19646) (1.59335,1.11309) (1.57704,0.937383) (1.57688,0.820749) (1.55832,0.647501) (1.54174,0.429295) (1.52409,0.275361) (1.50781,0.1631) (1.49245,0.0871384) (1.47754,0.0458278) (1.46287,0.0214473) };
          \addlegendentry{$b=200$};
          \addplot coordinates { (2.02539,4.3554) (1.9949,4.329) (1.93423,4.28816) (1.90417,4.28082) (1.904,4.26257) (1.84355,4.20521) (1.84303,4.17711) (1.81303,4.17014) (1.72142,4.14938) (1.69034,4.02253) (1.68165,3.89712) (1.66155,3.81971) (1.6415,3.73413) (1.63064,3.61272) (1.62082,3.48918) (1.60148,3.1407) (1.601,3.0175) (1.60077,2.61097) (1.58092,2.59336) (1.57071,2.17014) (1.56203,1.97628) (1.56186,1.49611) (1.54305,1.36986) (1.5235,0.828638) (1.5152,0.630915) (1.51484,0.581869) (1.5055,0.437063) (1.48839,0.202421) (1.48796,0.198342) (1.47146,0.0731764) (1.45484,0.0228444) };
          \addlegendentry{$b=500$};
          \addplot coordinates { (2.09326,4.43656) (2.06393,4.43262) (2.00179,4.42478) (1.9712,4.41306) (1.94208,4.363) (1.91107,4.3554) (1.85112,4.329) (1.85108,4.27716) (1.82072,4.2517) (1.8205,4.22297) (1.7803,4.2123) (1.69953,4.1632) (1.69926,4.11184) (1.69922,4.095) (1.69906,4.08163) (1.67916,4.05515) (1.63891,3.93082) (1.63861,3.9032) (1.61852,3.87597) (1.60906,3.87597) (1.60815,3.73134) (1.59747,3.65764) (1.57798,3.63636) (1.57795,3.51124) (1.55795,3.39674) (1.54766,3.00661) (1.53749,2.74123) (1.53731,2.57865) (1.51726,1.77873) (1.49859,0.381563) (1.48783,0.104106) (1.48769,0.101096) (1.47884,0.00597334) };
          \addlegendentry{$b=5000$};
          \addplot coordinates { (2.0202,3.8432) (1.97984,3.8373) (1.89829,3.79939) (1.83862,3.74251) (1.79761,3.71195) (1.77833,3.7037) (1.77773,3.64166) (1.75743,3.62845) (1.73705,3.56125) (1.71821,3.4626) (1.69752,3.4626) (1.67693,3.42231) (1.63655,3.39674) (1.61685,3.26158) (1.5961,3.13873) (1.59546,2.96209) (1.57574,2.91886) (1.55535,2.76549) (1.55511,2.49252) (1.53507,2.34412) (1.51506,1.50466) (1.49566,0.0429852) };
          \addlegendentry{$b=20000$};
          \addplot coordinates { (2.091,3.56633) (1.99908,3.54108) (1.93966,3.5137) (1.939,3.46981) (1.84752,3.465) (1.81711,3.4626) (1.78689,3.37154) (1.7565,3.32447) (1.72653,3.2113) (1.69637,3.2113) (1.66585,3.04507) (1.63624,2.88517) (1.61479,2.8393) (1.60561,2.72926) (1.59562,2.67094) (1.57543,2.53936) (1.57523,2.28415) (1.5747,2.28415) (1.57445,2.21435) (1.55512,2.13038) (1.55508,2.09468) (1.54504,2.08247) (1.5346,1.85667) (1.53422,1.55473) (1.5148,1.12133) (1.51471,1.07181) (1.51395,0.156357) };
          \addlegendentry{$b=100000$};
        \end{axis}
    \end{tikzpicture}

    \begin{tikzpicture}
        \ref*{legendParetoPlotsByBucketSize}
    \end{tikzpicture}
    \caption{Pareto plot over the construction throughput of different SicHash configurations by bucket size. Load factor $N/M=0.9$. Left: measured space usage. Right: hypothetical space usage, assuming that bucket metadata is encoded with Elias-Fano and Golomb-Rice codes.}
    \label{fig:paretoBucketSize}
\end{figure}
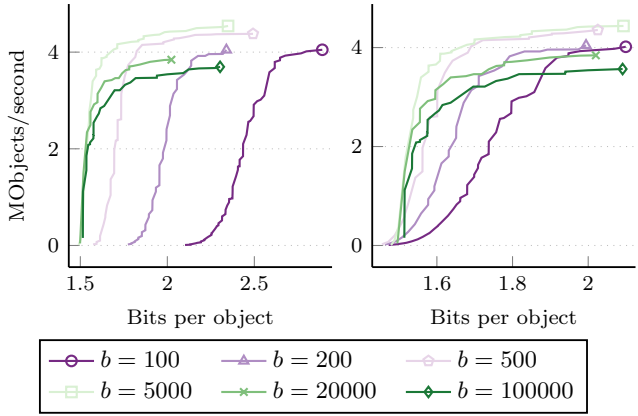

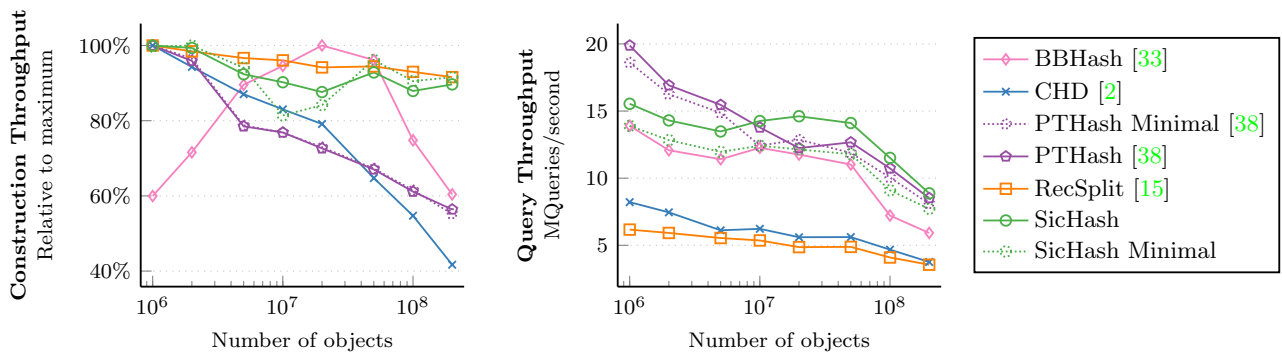
\begin{figure*}[t]
    \centering
    \begin{tikzpicture}
        \begin{axis}[
            plotMaximumLoadFactor,
            xlabel={Number of objects},
            ylabel={\begin{tabular}{c}\textbf{Construction Throughput}\\Relative to maximum\end{tabular}},
            xmode=log,
            yticklabel={\pgfmathprintnumber\tick\%},
            width=43mm,
          ]
          \addplot[mark=diamond,color=colorBbhash,solid] coordinates { (1000000,59.9664) (2000000,71.5909) (5000000,89.5685) (10000000,94.5779) (20000000,100.0) (50000000,96.1469) (1e+08,74.7947) (2e+08,60.3952) };
          \addlegendentry{BBHash \cite{limasset2017fast}};
          \addplot[mark=x,color=colorChd,solid] coordinates { (1000000,100.0) (2000000,94.3086) (5000000,87.0511) (10000000,83.0139) (20000000,79.1194) (50000000,64.7413) (1e+08,54.7227) (2e+08,41.6985) };
          \addlegendentry{CHD \cite{belazzougui2009hash}};
          \addplot[mark=pentagon,color=colorPthash,densely dotted] coordinates { (1000000,100.0) (2000000,96.4976) (5000000,78.7701) (10000000,76.8173) (20000000,72.9602) (50000000,67.2791) (1e+08,61.5176) (2e+08,55.2217) };
          \addlegendentry{PTHash Minimal \cite{pibiri2021pthash}};
          \addplot[mark=pentagon,color=colorPthash,solid] coordinates { (1000000,100.0) (2000000,95.9136) (5000000,78.4973) (10000000,76.8763) (20000000,72.6477) (50000000,67.0437) (1e+08,61.1326) (2e+08,56.3998) };
          \addlegendentry{PTHash \cite{pibiri2021pthash}};
          \addplot[mark=square,color=colorRecSplit,solid] coordinates { (1000000,100.0) (2000000,98.4375) (5000000,96.6593) (10000000,96.0476) (20000000,94.1784) (50000000,94.4373) (1e+08,92.9825) (2e+08,91.5884) };
          \addlegendentry{RecSplit \cite{esposito2020recsplit}};
          \addplot[mark=o,color=colorHeterogeneous,solid] coordinates { (1000000,100.0) (2000000,99.3519) (5000000,92.3514) (10000000,90.2287) (20000000,87.6204) (50000000,92.8393) (1e+08,87.8989) (2e+08,89.6528) };
          \addlegendentry{SicHash};
          \addplot[mark=o,color=colorHeterogeneous,densely dotted] coordinates { (1000000,99.7124) (2000000,100.0) (5000000,94.0893) (10000000,81.4154) (20000000,84.1424) (50000000,96.1467) (1e+08,90.5261) (2e+08,91.5348) };
          \addlegendentry{SicHash Minimal};

          \legend{};
        \end{axis}
    \end{tikzpicture}
    \hspace{3mm}
    \begin{tikzpicture}
        \begin{axis}[
            plotMaximumLoadFactor,
            xlabel={Number of objects},
            ylabel={\begin{tabular}{c}\textbf{Query Throughput}\\MQueries/second\end{tabular}},
            legend columns=1,
            legend style={at={(1.1,0.5)},anchor=west},
            xmode=log,
            width=43mm,
          ]
          \addplot[mark=diamond,color=colorBbhash,solid] coordinates { (1000000,13.9053) (2000000,12.0792) (5000000,11.4129) (10000000,12.2624) (20000000,11.7578) (50000000,11.0205) (1e+08,7.21605) (2e+08,5.92066) };
          \addlegendentry{BBHash \cite{limasset2017fast}};
          \addplot[mark=x,color=colorChd,solid] coordinates { (1000000,8.21558) (2000000,7.45453) (5000000,6.1182) (10000000,6.22704) (20000000,5.60255) (50000000,5.61293) (1e+08,4.67683) (2e+08,3.76053) };
          \addlegendentry{CHD \cite{belazzougui2009hash}};
          \addplot[mark=pentagon,color=colorPthash,densely dotted] coordinates { (1000000,18.6122) (2000000,16.2608) (5000000,14.8711) (10000000,12.4254) (20000000,12.8728) (50000000,11.8831) (1e+08,10.0529) (2e+08,8.09629) };
          \addlegendentry{PTHash Minimal \cite{pibiri2021pthash}};
          \addplot[mark=pentagon,color=colorPthash,solid] coordinates { (1000000,19.8853) (2000000,16.9192) (5000000,15.4634) (10000000,13.7779) (20000000,12.2429) (50000000,12.6721) (1e+08,10.7319) (2e+08,8.52079) };
          \addlegendentry{PTHash \cite{pibiri2021pthash}};
          \addplot[mark=square,color=colorRecSplit,solid] coordinates { (1000000,6.16941) (2000000,5.91973) (5000000,5.54324) (10000000,5.36164) (20000000,4.86547) (50000000,4.88234) (1e+08,4.094) (2e+08,3.55442) };
          \addlegendentry{RecSplit \cite{esposito2020recsplit}};
          \addplot[mark=o,color=colorHeterogeneous,solid] coordinates { (1000000,15.5412) (2000000,14.3021) (5000000,13.4807) (10000000,14.2579) (20000000,14.6135) (50000000,14.1084) (1e+08,11.5119) (2e+08,8.87206) };
          \addlegendentry{SicHash};
          \addplot[mark=o,color=colorHeterogeneous,densely dotted] coordinates { (1000000,13.8979) (2000000,12.8355) (5000000,11.9614) (10000000,12.4182) (20000000,12.1384) (50000000,11.7878) (1e+08,9.09753) (2e+08,7.71724) };
          \addlegendentry{SicHash Minimal};
        \end{axis}
    \end{tikzpicture}
    \caption{Comparison of different competitors by number of objects $N$. Left: Construction throughput relative to each method's maximum. Right: Query throughput. For the PHFs, the load factor is $0.95$. The parameters are selected from the Pareto front in a way that all PHFs achieve a space usage of 1.8 bits per object, BBHash achieves 3.5 bits per object, and all other MPHFs achieve a space usage of 2.2 bits per object.}
    \label{fig:scaling}
\end{figure*}

\begin{figure*}[h]
    \input{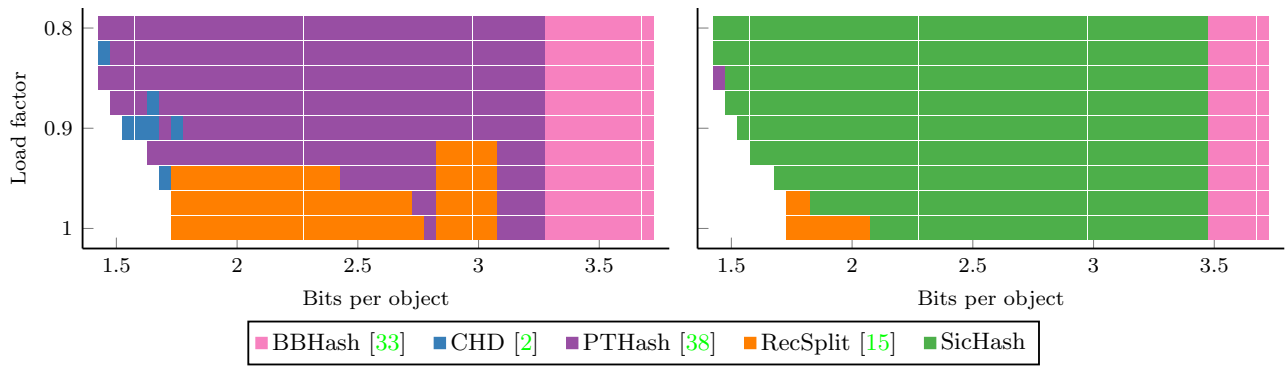}
    \caption{Heat map showing which PHF construction method has the fastest construction time for each given combination of load factor and space usage. Left: Algorithms without SicHash. Right: Algorithms including SicHash. SicHash covers most of the configurations.}
    \label{fig:heatmap}
\end{figure*}

\begin{figure*}[p]
    \input{fig/paretoConstruction}
    \caption{Space usage vs construction throughput for different load factors. Each line is a Pareto front showing configurations of a method that are not dominated by other configurations. Each plot contains a PHF that achieves \emph{at least} the specific load factor. In particular, minimal perfect hash functions are included in all plots. The top six plots use $N=5$~Million objects, while the bottom six plots use $N=20$~Million objects. Methods that are entirely uncompetitive are excluded from the plot with $N=20$~Million objects.}
    \label{fig:paretoConstruction}
\end{figure*}

\myparagraph{Bucket Size.}
For larger buckets, the relative overhead of encoding the metadata is reduced, but they can be overloaded less.
\Cref{fig:paretoBucketSize} (left) shows a Pareto plot\footnote{A configuration is on the Pareto front if it is not dominated by any other configuration with respect to both construction time and space consumption.} of SicHash configurations using different bucket sizes and indicates that a bucket size of $b=5000$ is optimal, which is why we choose that parameter in all other measurements.
\Cref{fig:paretoBucketSize} (right) shows hypothetical values for the space usage when assuming that the per-bucket metadata is encoded with Elias-Fano and Golomb-Rice coding instead of arrays (see \cref{sec:theDataStructure}).

\myparagraph{Parameter Choices.}
A guideline on how to choose parameters for an efficient PHF construction can be determined by running a benchmark of multiple configurations and determining the Pareto front.
The table in \cref{sec:parameterChoices} give exemplary configuration parameters $p_1, p_2$ that are on the Pareto front on our machine.
As discussed in \cref{sec:loadFactors}, an ordinary cuckoo hash table leads to a PHF with a load factor of $97.68\%$ and $2$ bits per object.
Our implementation of SicHash achieves the same load factor using only $1.84$ bits per object ($p_1=49\%$, $p_2=22\%$).
When constructing without overloading, placing the objects in the cuckoo hash tables takes about $50\%$ of the time and constructing the retrieval data structure again takes about $50\%$ of the time.
When increasing the amount of overloading, placing objects in the cuckoo hash table takes longer because more retries are needed.

\subsection{Comparison with Competitors.}
For demonstrating the performance of SicHash, we compare it to competitors from the literature (see \cref{sec:relworkPerfectHashing}).
From the cmph library, we include CHD \cite{belazzougui2009hash}, BDZ \cite{botelho2013practical}, BMZ\footref{fn:mphf} \cite{botelho2004new}, and FCH\footref{fn:mphf} \cite{fox1992faster}.
Additionally, we include WBPM\footref{fn:mphf} \cite{weaver2020constructing}, PTHash \cite{pibiri2021pthash}, RecSplit\footref{fn:mphf} \cite{esposito2020recsplit}, and BBHash\footref{fn:mphf} \cite{limasset2017fast}.
\addtocounter{footnote}{1}
\footnotetext{\label{fn:mphf}Method only supports construction of MPHF, not PHF. Included in plots for all load factors.}

\myparagraph{Construction Scaling.}
\Cref{fig:scaling} (left) shows how the construction throughput scales with $N$.
The configuration parameters of each method are selected from the Pareto front in a way that all PHFs with a load factor of $0.95$ need 1.8 bits per object, BBHash needs 4.0 bits, and all other MPHFs need 2.2 bits per object.
To be more fair about the different load factors and space usage configurations, the construction throughput is given \emph{relative} to each competitor's maximum performance.
The construction throughput of PTHash, BBHash and CHD drops significantly when increasing $N$, while SicHash and RecSplit stay more constant.
This can be explained by the fact that SicHash and RecSplit hash objects to buckets of expected constant size that are then constructed independently.
The other competitors, in contrast, access an array with size proportional to $N$, which leads to lower cache locality.
While putting an additional layer on top of the data structures for more cache locality is always possible, this would affect query time and space usage.

\myparagraph{Query Scaling.}
The query throughput in \cref{fig:scaling} drops when increasing $N$ because the data structure gets larger and can be cached less.
One of the main design goals of PTHash are its fast queries.
SicHash comes close to PTHash or even surpasses it in terms of query throughput, while having a more favorable trade-off between space usage and construction performance.
The minimal perfect variant has a query overhead of about 5--10\%.

\myparagraph{Trade-Off.}
For each competitor, we run experiments with a number of configurations to provide a trade-off between construction time and space usage.
\Cref{fig:paretoConstruction} gives Pareto fronts of our method and competitors for different load factors and input sizes $N$.
SicHash dominates the Pareto front for load factors from $0.85$ to $0.97$.
While there are better competitors for very small or very large data structures, SicHash covers a wide range of configurations (see also \cref{fig:heatmap}).
For some space requirements, SicHash can be constructed up to 3 times faster than the next best competitor.
For small $N \leq 1$~Million, the construction performance of SicHash and the closest competitors is comparable.
Refer to \cref{sec:additionalExperiments} for details.
For MPHFs, RecSplit produces the smallest data structures, but SicHash is competitive for a range of configurations while having significantly faster queries (see \Cref{fig:scaling}).

\section{Conclusion and Future Work}
\label{sec:conclusion}
With SicHash, we present a new perfect hash function, which places objects in a number of small, irregular cuckoo hash tables.
Making the tables small enables \emph{overloading}, which achieves higher load factors than the asymptotic bound.
Using irregular cuckoo hashing enables fine-grained control over the load factors and lower space usage.
We then use space efficient retrieval data structures to store the final placement.
Our implementation improves the state of the art in perfect hash functions for a wide range of load factors and space usage configurations.

Future work might include a parallel construction algorithm and construction using a different insertion strategy than random walk.
From a theoretical point, it would be interesting to look at construction success probabilities for overloaded cuckoo hash tables.

\subsection*{Acknowledgements.}
The authors would like to thank Martin Dietzfelbinger for early discussions leading to this paper.
This project has received funding from the European Research Council (ERC) under the European Union’s Horizon 2020 research and innovation programme (grant agreement No. 882500).

\includegraphics[width=4cm]{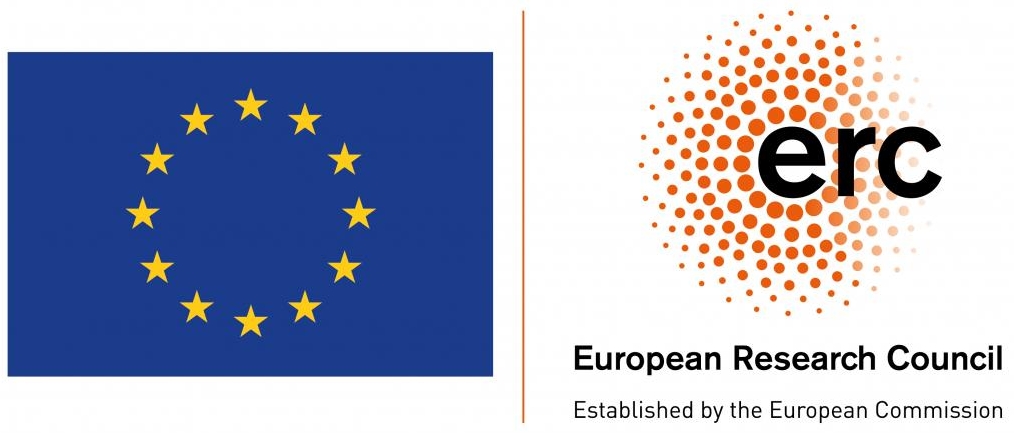}

\clearpage
\bibliographystyle{plainurl}
\bibliography{paper}

\begin{thebibliography}{10}

\bibitem{bast2006matching}
Holger Bast, Kurt Mehlhorn, Guido Schafer, and Hisao Tamaki.
\newblock Matching algorithms are fast in sparse random graphs.
\newblock {\em Theory of Computing Systems}, 39(1):3--14, 2006.

\bibitem{belazzougui2009hash}
Djamal Belazzougui, Fabiano~C. Botelho, and Martin Dietzfelbinger.
\newblock Hash, displace, and compress.
\newblock In {\em {ESA}}, volume 5757 of {\em Lecture Notes in Computer
  Science}, pages 682--693. Springer, 2009.
\newblock \href {https://doi.org/10.1007/978-3-642-04128-0_61}
  {\path{doi:10.1007/978-3-642-04128-0_61}}.

\bibitem{botelho2004new}
Fabiano~C Botelho, David~M Gomes, and Nivio Ziviani.
\newblock A new algorithm for constructing minimal perfect hash functions.
\newblock {\em differences}, 100(2):09, 2004.

\bibitem{botelho2007simple}
Fabiano~C. Botelho, Rasmus Pagh, and Nivio Ziviani.
\newblock Simple and space-efficient minimal perfect hash functions.
\newblock In {\em {WADS}}, volume 4619 of {\em Lecture Notes in Computer
  Science}, pages 139--150. Springer, 2007.
\newblock \href {https://doi.org/10.1007/978-3-540-73951-7_13}
  {\path{doi:10.1007/978-3-540-73951-7_13}}.

\bibitem{botelho2013practical}
Fabiano~C. Botelho, Rasmus Pagh, and Nivio Ziviani.
\newblock Practical perfect hashing in nearly optimal space.
\newblock {\em Inf. Syst.}, 38(1):108--131, 2013.
\newblock \href {https://doi.org/10.1016/J.IS.2012.06.002}
  {\path{doi:10.1016/J.IS.2012.06.002}}.

\bibitem{brain1990perfect}
Marshall~D. Brain and Alan~L. Tharp.
\newblock Perfect hashing using sparse matrix packing.
\newblock {\em Inf. Syst.}, 15(3):281--290, 1990.
\newblock \href {https://doi.org/10.1016/0306-4379(90)90001-6}
  {\path{doi:10.1016/0306-4379(90)90001-6}}.

\bibitem{bruck2006weighted}
Jehoshua Bruck, Jie Gao, and Anxiao Jiang.
\newblock Weighted bloom filter.
\newblock In {\em {ISIT}}, pages 2304--2308. {IEEE}, 2006.
\newblock \href {https://doi.org/10.1109/ISIT.2006.261978}
  {\path{doi:10.1109/ISIT.2006.261978}}.

\bibitem{CSLSR11}
Jarrod~A. Chapman, Isaac Ho, Sirisha Sunkara, Shujun Luo, Gary~P. Schroth, and
  Daniel~S. Rokhsar.
\newblock Meraculous: De novo genome assembly with short paired-end reads.
\newblock {\em PLOS ONE}, 6(8):1--13, 08 2011.
\newblock \href {https://doi.org/10.1371/journal.pone.0023501}
  {\path{doi:10.1371/journal.pone.0023501}}.

\bibitem{czech1992optimal}
Zbigniew~J. Czech, George Havas, and Bohdan~S. Majewski.
\newblock An optimal algorithm for generating minimal perfect hash functions.
\newblock {\em Inf. Process. Lett.}, 43(5):257--264, 1992.
\newblock \href {https://doi.org/10.1016/0020-0190(92)90220-P}
  {\path{doi:10.1016/0020-0190(92)90220-P}}.

\bibitem{dietzfelbinger2010tight}
Martin Dietzfelbinger, Andreas Goerdt, Michael Mitzenmacher, Andrea Montanari,
  Rasmus Pagh, and Michael Rink.
\newblock Tight thresholds for cuckoo hashing via {XORSAT}.
\newblock In {\em {ICALP} {(1)}}, volume 6198 of {\em Lecture Notes in Computer
  Science}, pages 213--225. Springer, 2010.
\newblock \href {https://doi.org/10.1007/978-3-642-14165-2_19}
  {\path{doi:10.1007/978-3-642-14165-2_19}}.

\bibitem{dietzfelbinger2011cuckoo}
Martin Dietzfelbinger, Michael Mitzenmacher, and Michael Rink.
\newblock Cuckoo hashing with pages.
\newblock In {\em {ESA}}, volume 6942 of {\em Lecture Notes in Computer
  Science}, pages 615--627. Springer, 2011.
\newblock \href {https://doi.org/10.1007/978-3-642-23719-5_52}
  {\path{doi:10.1007/978-3-642-23719-5_52}}.

\bibitem{dietzfelbinger2007balanced}
Martin Dietzfelbinger and Christoph Weidling.
\newblock Balanced allocation and dictionaries with tightly packed constant
  size bins.
\newblock In {\em {ICALP}}, volume 3580 of {\em Lecture Notes in Computer
  Science}, pages 166--178. Springer, 2005.
\newblock \href {https://doi.org/10.1007/11523468_14}
  {\path{doi:10.1007/11523468_14}}.

\bibitem{dillinger2021fast}
Peter~C. Dillinger, Lorenz H{\"{u}}bschle{-}Schneider, Peter Sanders, and
  Stefan Walzer.
\newblock Fast succinct retrieval and approximate membership using ribbon.
\newblock In {\em {SEA}}, volume 233 of {\em LIPIcs}, pages 4:1--4:20. Schloss
  Dagstuhl - Leibniz-Zentrum f{\"{u}}r Informatik, 2022.
\newblock \href {https://doi.org/10.4230/LIPICS.SEA.2022.4}
  {\path{doi:10.4230/LIPICS.SEA.2022.4}}.

\bibitem{Elias74}
Peter Elias.
\newblock Efficient storage and retrieval by content and address of static
  files.
\newblock {\em J. {ACM}}, 21(2):246--260, 1974.
\newblock \href {https://doi.org/10.1145/321812.321820}
  {\path{doi:10.1145/321812.321820}}.

\bibitem{esposito2020recsplit}
Emmanuel Esposito, Thomas~Mueller Graf, and Sebastiano Vigna.
\newblock Recsplit: Minimal perfect hashing via recursive splitting.
\newblock In {\em {ALENEX}}, pages 175--185. {SIAM}, 2020.
\newblock \href {https://doi.org/10.1137/1.9781611976007.14}
  {\path{doi:10.1137/1.9781611976007.14}}.

\bibitem{Fano71}
Robert~Mario Fano.
\newblock On the number of bits required to implement an associative memory.
\newblock Technical report, MIT, Computer Structures Group, 1971.
\newblock Project MAC, Memorandum 61".

\bibitem{fotakis2005space}
Dimitris Fotakis, Rasmus Pagh, Peter Sanders, and Paul~G. Spirakis.
\newblock Space efficient hash tables with worst case constant access time.
\newblock {\em Theory Comput. Syst.}, 38(2):229--248, 2005.
\newblock \href {https://doi.org/10.1007/S00224-004-1195-X}
  {\path{doi:10.1007/S00224-004-1195-X}}.

\bibitem{fountoulakis2012sharp}
Nikolaos Fountoulakis and Konstantinos Panagiotou.
\newblock Sharp load thresholds for cuckoo hashing.
\newblock {\em Random Struct. Algorithms}, 41(3):306--333, 2012.
\newblock \href {https://doi.org/10.1002/RSA.20426}
  {\path{doi:10.1002/RSA.20426}}.

\bibitem{fountoulakis2013insertion}
Nikolaos Fountoulakis, Konstantinos Panagiotou, and Angelika Steger.
\newblock On the insertion time of cuckoo hashing.
\newblock {\em {SIAM} J. Comput.}, 42(6):2156--2181, 2013.
\newblock \href {https://doi.org/10.1137/100797503}
  {\path{doi:10.1137/100797503}}.

\bibitem{fox1992faster}
Edward~A. Fox, Qi~Fan Chen, and Lenwood~S. Heath.
\newblock A faster algorithm for constructing minimal perfect hash functions.
\newblock In {\em {SIGIR}}, pages 266--273. {ACM}, 1992.
\newblock \href {https://doi.org/10.1145/133160.133209}
  {\path{doi:10.1145/133160.133209}}.

\bibitem{frieze2009analysis}
Alan~M. Frieze, P{\'{a}}ll Melsted, and Michael Mitzenmacher.
\newblock An analysis of random-walk cuckoo hashing.
\newblock In {\em {APPROX-RANDOM}}, volume 5687 of {\em Lecture Notes in
  Computer Science}, pages 490--503. Springer, 2009.
\newblock \href {https://doi.org/10.1007/978-3-642-03685-9_37}
  {\path{doi:10.1007/978-3-642-03685-9_37}}.

\bibitem{gog2014theory}
Simon Gog, Timo Beller, Alistair Moffat, and Matthias Petri.
\newblock From theory to practice: Plug and play with succinct data structures.
\newblock In {\em {SEA}}, volume 8504 of {\em Lecture Notes in Computer
  Science}, pages 326--337. Springer, 2014.
\newblock \href {https://doi.org/10.1007/978-3-319-07959-2_28}
  {\path{doi:10.1007/978-3-319-07959-2_28}}.

\bibitem{golomb1966run}
Solomon~W. Golomb.
\newblock Run-length encodings (corresp.).
\newblock {\em {IEEE} Trans. Inf. Theory}, 12(3):399--401, 1966.
\newblock \href {https://doi.org/10.1109/TIT.1966.1053907}
  {\path{doi:10.1109/TIT.1966.1053907}}.

\bibitem{hopcroft1973n}
John~E. Hopcroft and Richard~M. Karp.
\newblock An n\({}^{\mbox{5/2}}\) algorithm for maximum matchings in bipartite
  graphs.
\newblock {\em {SIAM} J. Comput.}, 2(4):225--231, 1973.

\bibitem{jaeschke1981reciprocal}
Gerhard Jaeschke.
\newblock Reciprocal hashing: {A} method for generating minimal perfect hashing
  functions.
\newblock {\em Commun. {ACM}}, 24(12):829--833, 1981.
\newblock \href {https://doi.org/10.1145/358800.358806}
  {\path{doi:10.1145/358800.358806}}.

\bibitem{khosla2013balls}
Megha Khosla.
\newblock Balls into bins made faster.
\newblock In {\em {ESA}}, volume 8125 of {\em Lecture Notes in Computer
  Science}, pages 601--612. Springer, 2013.
\newblock \href {https://doi.org/10.1007/978-3-642-40450-4_51}
  {\path{doi:10.1007/978-3-642-40450-4_51}}.

\bibitem{khosla2019faster}
Megha Khosla and Avishek Anand.
\newblock A faster algorithm for cuckoo insertion and bipartite matching in
  large graphs.
\newblock {\em Algorithmica}, 81(9):3707--3724, 2019.
\newblock \href {https://doi.org/10.1007/S00453-019-00595-4}
  {\path{doi:10.1007/S00453-019-00595-4}}.

\bibitem{kurpicz2022pasta}
Florian Kurpicz.
\newblock Engineering compact data structures for rank and select queries on
  bit vectors.
\newblock In {\em {SPIRE}}, volume 13617 of {\em Lecture Notes in Computer
  Science}, pages 257--272. Springer, 2022.
\newblock \href {https://doi.org/10.1007/978-3-031-20643-6\_19}
  {\path{doi:10.1007/978-3-031-20643-6\_19}}.

\bibitem{kuszmaul2016kickout}
William Kuszmaul.
\newblock Fast concurrent cuckoo kick-out eviction schemes for high-density
  tables.
\newblock {\em CoRR}, abs/1605.05236, 2016.
\newblock \href {https://doi.org/10.48550/arXiv.1605.05236}
  {\path{doi:10.48550/arXiv.1605.05236}}.

\bibitem{kutzelnigg2010improved}
Reinhard Kutzelnigg.
\newblock An improved version of cuckoo hashing: Average case analysis of
  construction cost and search operations.
\newblock In {\em {IWOCA}}, pages 253--266. College Publications, 2008.

\bibitem{lehman20093}
Eric Lehman and Rina Panigrahy.
\newblock 3.5-way cuckoo hashing for the price of 2-and-a-bit.
\newblock In {\em {ESA}}, volume 5757 of {\em Lecture Notes in Computer
  Science}, pages 671--681. Springer, 2009.
\newblock \href {https://doi.org/10.1007/978-3-642-04128-0_60}
  {\path{doi:10.1007/978-3-642-04128-0_60}}.

\bibitem{L:A_New_Approach:2012}
Marc Lelarge.
\newblock A new approach to the orientation of random hypergraphs.
\newblock In {\em Proc. 23rd SODA}, pages 251--264. {SIAM}, 2012.
\newblock \href {https://doi.org/10.1137/1.9781611973099.23}
  {\path{doi:10.1137/1.9781611973099.23}}.

\bibitem{limasset2017fast}
Antoine Limasset, Guillaume Rizk, Rayan Chikhi, and Pierre Peterlongo.
\newblock Fast and scalable minimal perfect hashing for massive key sets.
\newblock In {\em {SEA}}, volume~75 of {\em LIPIcs}, pages 25:1--25:16. Schloss
  Dagstuhl - Leibniz-Zentrum f{\"{u}}r Informatik, 2017.
\newblock \href {https://doi.org/10.4230/LIPICS.SEA.2017.25}
  {\path{doi:10.4230/LIPICS.SEA.2017.25}}.

\bibitem{luby2001efficient}
Michael Luby, Michael Mitzenmacher, Mohammad~Amin Shokrollahi, and Daniel~A.
  Spielman.
\newblock Efficient erasure correcting codes.
\newblock {\em {IEEE} Trans. Inf. Theory}, 47(2):569--584, 2001.
\newblock \href {https://doi.org/10.1109/18.910575}
  {\path{doi:10.1109/18.910575}}.

\bibitem{MSSZ14}
Ingo M{\"u}ller, Peter Sanders, Robert Schulze, and Wei Zhou.
\newblock Retrieval and perfect hashing using fingerprinting.
\newblock In {\em 13th Symposium on Experimental Algorithms (SEA)}, volume 8504
  of {\em LNCS}, pages 138--149. Springer, 2014.
\newblock \href {https://doi.org/10.1007/978-3-319-07959-2_12}
  {\path{doi:10.1007/978-3-319-07959-2_12}}.

\bibitem{navarro2016compact}
Gonzalo Navarro.
\newblock {\em Compact Data Structures - {A} Practical Approach}.
\newblock Cambridge University Press, 2016.

\bibitem{pagh2004cuckoo}
Rasmus Pagh and Flemming~Friche Rodler.
\newblock Cuckoo hashing.
\newblock {\em J. Algorithms}, 51(2):122--144, 2004.
\newblock \href {https://doi.org/10.1016/j.jalgor.2003.12.002}
  {\path{doi:10.1016/j.jalgor.2003.12.002}}.

\bibitem{pibiri2021pthash}
Giulio~E. Pibiri and Roberto Trani.
\newblock {PTHash}: Revisiting {FCH} minimal perfect hashing.
\newblock In {\em {SIGIR}}, pages 1339--1348. {ACM}, 2021.
\newblock \href {https://doi.org/10.1145/3404835.3462849}
  {\path{doi:10.1145/3404835.3462849}}.

\bibitem{rice1979some}
Robert~F. Rice.
\newblock Some practical universal noiseless coding techniques.
\newblock {\em Jet Propulsion Laboratory, JPL Publication}, 1979.

\bibitem{SMDD19}
Peter Sanders, Kurt Mehlhorn, Martin Dietzfelbinger, and Roman Dementiev.
\newblock {\em Sequential and Parallel Algorithms and Data Structures - The
  Basic Toolbox}.
\newblock Springer, 2019.
\newblock \href {https://doi.org/10.1007/978-3-030-25209-0}
  {\path{doi:10.1007/978-3-030-25209-0}}.

\bibitem{sprugnoli1977perfect}
Renzo Sprugnoli.
\newblock Perfect hashing functions: {A} single probe retrieving method for
  static sets.
\newblock {\em Commun. {ACM}}, 20(11):841--850, 1977.
\newblock \href {https://doi.org/10.1145/359863.359887}
  {\path{doi:10.1145/359863.359887}}.

\bibitem{walzer2017load}
Stefan Walzer.
\newblock Load thresholds for cuckoo hashing with overlapping blocks.
\newblock In {\em {ICALP}}, volume 107 of {\em LIPIcs}, pages 102:1--102:10.
  Schloss Dagstuhl - Leibniz-Zentrum f{\"{u}}r Informatik, 2018.
\newblock \href {https://doi.org/10.4230/LIPICS.ICALP.2018.102}
  {\path{doi:10.4230/LIPICS.ICALP.2018.102}}.

\bibitem{walzer2021peeling}
Stefan Walzer.
\newblock Peeling close to the orientability threshold - spatial coupling in
  hashing-based data structures.
\newblock In {\em {SODA}}, pages 2194--2211. {SIAM}, 2021.
\newblock \href {https://doi.org/10.1137/1.9781611976465.131}
  {\path{doi:10.1137/1.9781611976465.131}}.

\bibitem{walzer2022insertion}
Stefan Walzer.
\newblock Insertion time of random walk cuckoo hashing below the peeling
  threshold.
\newblock {\em CoRR}, abs/2202.05546, 2022.

\bibitem{weaver2020constructing}
Sean~A. Weaver and Marijn Heule.
\newblock Constructing minimal perfect hash functions using {SAT} technology.
\newblock In {\em {AAAI}}, pages 1668--1675. {AAAI} Press, 2020.

\end{thebibliography}

\clearpage
\appendix
\onecolumn

\section{Overloading Binary Cuckoo Hash Tables}
\label{sec:overloadBinary}
The effect of overloading cuckoo hash tables is strongly visible for ordinary tables that use only $d=2$ hash functions.
\Cref{fig:overloadBinary} gives box plots of the achieved load factors with different table sizes $M$.

\begin{figure*}[h]
    \input{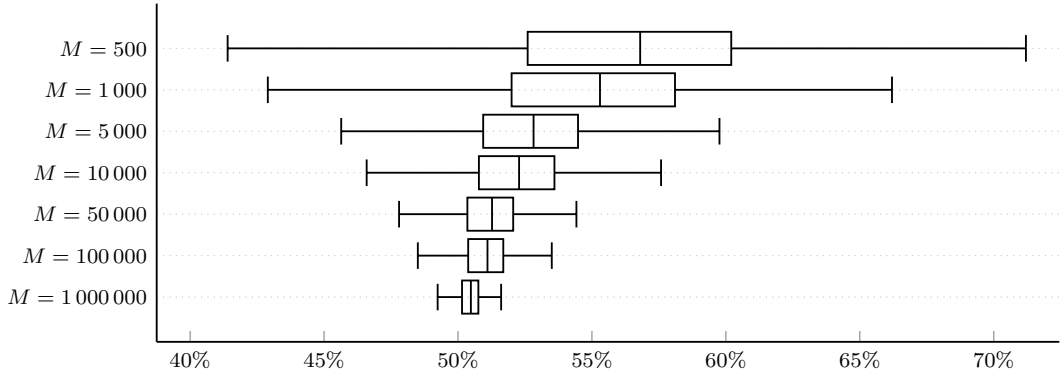}
    \caption{Achived load factors when constructing ordinary, binary cuckoo hash tables of different sizes.}
    \label{fig:overloadBinary}
\end{figure*}

\vspace{-5mm}
\section{Additional Experimental Data}
\label{sec:additionalExperiments}
\Cref{fig:paretoConstruction1M} gives measurements of the construction throughput for $N=1$~Million Objects.
SicHash is still competitive but does not dominate the Pareto front like for larger $N$.
\Cref{fig:paretoConstructionAdditional} gives the comparison with competitors on two additional machines.
The overall behavior is similar to the Intel machine that we run the main experiments on.

\begin{figure*}[h]
    \input{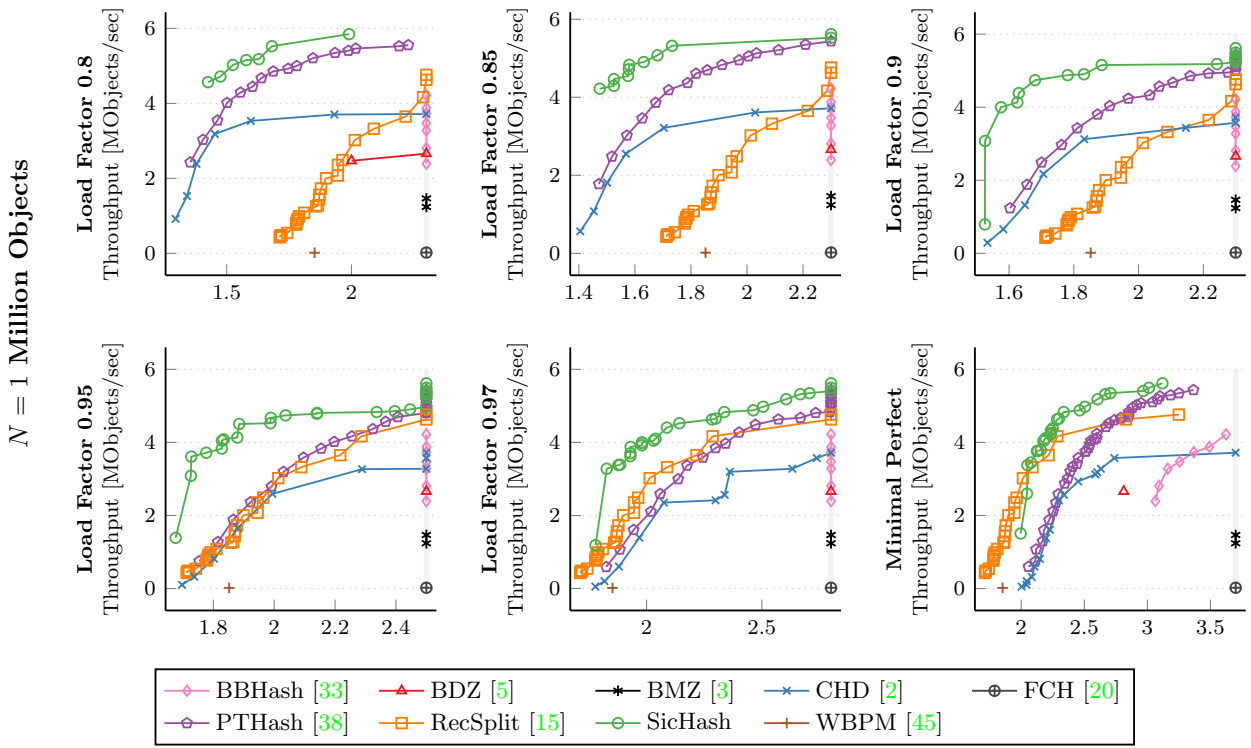}
    \vspace{-3mm}
    \caption{Space usage vs construction throughput for small $N=1$~Million.}
    \label{fig:paretoConstruction1M}
\end{figure*}

\begin{figure*}[p]
    \input{fig/paretoConstructionAdditional}
    \caption{Space usage vs construction throughput on two additional machines. Uses $N=5$~Million objects.}
    \label{fig:paretoConstructionAdditional}
\end{figure*}

\clearpage
\section{Parameter Choices}
\label{sec:parameterChoices}
\Cref{tab:parametersOverview} lists some configuration parameters that are on the Pareto front in \cref{fig:paretoConstruction}.

\begin{table*}[h]
  \caption{Selected configuration parameters from the pareto front (see \cref{fig:paretoConstruction}).}
  \label{tab:parametersOverview}
    \centering
  \addtolength\tabcolsep{-4pt}
  \begin{tabular}[t]{l l l r r }
    \toprule
                                                            & B/Obj & Method                          & MObj/s & ns/Obj \\ \midrule
      \multirow{6}{*}{\rotatebox[origin=c]{90}{$N/M$=1.0}}  & 2.01  & RecSplit, l=5 b=100             & 2.85   &  350   \\
                                                            & 2.06  & PTHash, lf=0.95, c=3.0          & 0.30   & 3333   \\
                                                            & 2.04  & lf=0.97, $p_1$=45\%, $p_2$=31\% & 2.94   &  340   \\ \cmidrule{2-5}
                                                            & 2.50  & RecSplit, l=3, b=50             & 3.86   &  259   \\
                                                            & 2.51  & PTHash, lf=0.95, c=6.0          & 3.06   &  327   \\
                                                            & 2.50  & lf=0.9, $p_1$=21\%, $p_2$=78\%  & 4.91   &  204   \\ \midrule
      \multirow{6}{*}{\rotatebox[origin=c]{90}{$N/M$=0.95}} & 1.78  & RecSplit, l=8, b=250            & 0.76   & 1316   \\
                                                            & 1.76  & PTHash, c=3.0                   & 0.61   & 1639   \\
                                                            & 1.75  & $p_1$=44\%, $p_2$=41\%          & 3.27   &  305   \\ \cmidrule{2-5}
                                                            & 2.09  & RecSplit, l=5, b=50             & 3.17   &  315   \\
                                                            & 2.06  & PTHash, c=5.8                   & 2.67   &  375   \\
                                                            & 2.04  & $p_1$=43\%, $p_2$=18\%          & 4.55   &  220   \\ \midrule
      \multirow{4}{*}{\rotatebox[origin=c]{90}{$N/M$=0.85}} & 1.61  & PTHash, c=4.6                   & 2.85   &  351   \\
                                                            & 1.61  & $p_1$=48\%, $p_2$=51\%          & 4.74   &  211   \\ \cmidrule{2-5}
                                                            & 1.91  & PTHash, c=7.4                   & 4.25   &  235   \\
                                                            & 1.89  & $p_1$=45\%, $p_2$=31\%          & 5.30   &  189   \\
      \bottomrule
  \end{tabular}
\end{table*}

\end{document}